\documentclass[lettersize,journal]{IEEEtran}
\usepackage{amsthm}
\usepackage{amsmath,amsfonts}
\usepackage{array}
\usepackage[caption=false,font=normalsize,labelfont=sf,textfont=sf]{subfig}
\usepackage{textcomp}
\usepackage{stfloats}
\usepackage{url}
\usepackage{verbatim}
\usepackage{graphicx}
\usepackage{stfloats}
\usepackage{multirow}
\usepackage{multicol}
\usepackage{cite}
\usepackage{algorithm,algpseudocode}
\usepackage{array} 
\usepackage{makecell}
\usepackage{xcolor}

\algnewcommand{\Inputs}[1]{%
  \State \textbf{Inputs:}
  \Statex \hspace*{\algorithmicindent}\parbox[t]{.8\linewidth}{\raggedright #1}
}
\algnewcommand{\Initialize}[1]{%
  \State \textbf{Initialization:}
  \Statex \hspace*{\algorithmicindent}\parbox[t]{.8\linewidth}{\raggedright #1}
}

\newtheorem{theorem}{Theorem}

\allowdisplaybreaks[3]

\def\BibTeX{{\rm B\kern-.05em{\sc i\kern-.025em b}\kern-.08em
    T\kern-.1667em\lower.7ex\hbox{E}\kern-.125emX}}
\usepackage{balance}
\begin{document}
\title{Electromagnetic Hybrid Beamforming for Holographic MIMO Communications}
\author{Ran Ji,~\IEEEmembership{Graduate Student Member,~IEEE}, Chongwen~Huang,~\IEEEmembership{Member,~IEEE}, Xiaoming Chen,~\IEEEmembership{Senior Member,~IEEE}, Wei E. I. Sha,~\IEEEmembership{Senior Member,~IEEE}, Linglong Dai,~\IEEEmembership{Fellow,~IEEE}, Jiguang He,~\IEEEmembership{Senior Member,~IEEE},  Zhaoyang Zhang,~\IEEEmembership{Senior Member,~IEEE}, Chau~Yuen,~\IEEEmembership{Fellow,~IEEE},\\
and M\'{e}rouane~Debbah,~\IEEEmembership{Fellow,~IEEE}
\thanks{
  The work was supported in part by the China National Key R\&D Program under Grant 2021YFA1000500 and 2023YFB2904804; in part by the National Natural Science Foundation of China under Grant 62331023, 62101492, 62394292 and U20A20158; in part by the Zhejiang Provincial Natural Science Foundation of China under Grant LR22F010002; in part by the Zhejiang Provincial Science and Technology Plan Project under Grant 2024C01033, in part by the Zhejiang University Global Partnership Fund; and in part by the Ministry of Education (MOE), Singapore, through its MOE Tier 2 (Award number MOE-T2EP50220-0019) and A\*STAR (Agency for Science, Technology and Research) Singapore, under Grant No. M22L1b0110. \textit{(Corresponding author: Chongwen Huang.)}
\par R. Ji and C. Huang are with College of Information Science and Electronic Engineering, Zhejiang University, Hangzhou 310027, China, with the State Key Laboratory of Integrated Service Networks, Xidian University, Xi’an 710071, China, and Zhejiang Provincial Key Laboratory of Info. Proc., Commun. \& Netw. (IPCAN), Hangzhou 310027, China. 
(E-mails: \{ranji, chongwenhuang\}@zju.edu.cn). 
\par Z. Zhang is with College of Information Science and Electronic Engineering, Zhejiang University, Hangzhou 310027, China, and also with Zhejiang Provincial Key Laboratory of Info. Proc., Commun. \& Netw. (IPCAN), Hangzhou 310027, China. (E-mail: zhzy@zju.edu.cn). 
\par W. E. I. Sha is with College of Information Science and Electronic Engineering, Zhejiang University, Hangzhou 310027, China. (E-mail: weisha@zju.edu.cn). 
\par X. Chen is with School of Information and Communications Engineering, Xi'an Jiaotong University, Xi'an 710049, China. (E-mail: xiaoming.chen@mail.xjtu.edu.cn). 
\par L. Dai is with the Department of Electronic Engineering, Tsinghua University, Beijing 100084, China. (E-mail: daill@tsinghua.edu.cn). 
\par J. He is with the Technology Innovation Institute, 9639 Masdar City, Abu Dhabi, United Arab Emirates. (E-mail: jiguang.he@tii.ae). 
\par C. Yuen is with the School of Electrical and Electronics Engineering, Nanyang Technological University, Singapore 639798, Singapore. (E-mail: chau.yuen@ntu.edu.sg). 
\par M. Debbah is with KU 6G Research Center, Khalifa University of Science and Technology, PO Box 127788, Abu Dhabi, United Arab Emirates. (E-mail: merouane.debbah@ku.ac.ae).}
}


\maketitle

%

\begin{abstract}

It is well known that there is inherent radiation pattern distortion for the commercial base station antenna array, which usually needs three antenna sectors to cover the whole space.
To eliminate pattern distortion and further enhance beamforming performance, we propose an electromagnetic hybrid beamforming (EHB) scheme based on a three-dimensional (3D) superdirective holographic antenna array. 
Specifically, EHB consists of antenna excitation current vectors (analog beamforming) and digital precoding matrices, where the implementation of analog beamforming involves the real-time adjustment of the radiation pattern to adapt it to the dynamic wireless environment. Meanwhile, the digital beamforming is optimized based on the channel characteristics of analog beamforming to further improve the achievable rate of communication systems.
An electromagnetic channel model incorporating array radiation patterns and the mutual coupling effect is also developed to evaluate the benefits of our proposed scheme. 
Simulation results demonstrate that our proposed EHB scheme with a 3D holographic array achieves a relatively flat superdirective beamforming gain and allows for programmable focusing directions throughout the entire spatial domain.
Furthermore, they also verify that the proposed scheme achieves a sum rate gain of over 150\% compared to traditional beamforming algorithms.

\end{abstract}


\section{Introduction}
5G communications have greatly benefited from the massive multiple-input-multiple-output (MIMO) and beamforming technologies, which have been proven to be effective in enhancing communication rates and reliability, leading to significant advancements compared with 4G communications. 
However, the radiation pattern of massive MIMO antenna arrays in commercial base stations has a narrow angular spread of approximately ±60 degrees \cite[Fig. 4.26]{Balanis2016}, which usually have the inherent radiation pattern distortion (i.e., radiation pattern that is non-uniform in the angular domain), especially for the directions of edge users, where the quality of service is relatively poor. In other words, a typical commercial antenna array usually exhibits a radiation gain loss of more than 4 dB (up to 20 dB especially for fixed broadside beamforming) in the edge directions, i.e., 60-degree direction, compared with the 0-degree (normal direction) \cite[Fig. 6.7]{Balanis2016}.
It is natural to think if it is possible to achieve a consistent and high-gain radiation pattern for users in all directions by leveraging a programmable radiation pattern and beamforming approaches.

\subsection{Prior Works}
Hybrid beamforming in millimeter wave or terahertz communication systems can partially address this issue, but its primary objectives remain to reduce hardware cost and optimize both the digital and analog beamforming matrices to approach the performance of digital solutions. 
Specifically, Huang \textit{et al.}\cite{xiaojinghuang1} proposed a hybrid beamforming architecture that combines low-cost analog radio-frequency (RF) chains with digital baseband precoding. Subsequently, Orthogonal Matching Pursuit (OMP) was employed for sparse channel reconstruction\cite{A.Alkhateeb1}.
Moreover, corresponding channel estimation and beamforming algorithm designs based on a hierarchical multi-resolution codebook were also provided. Yu \textit{et al.} \cite{XYu} evaluated the performance of fully-connected and partially-connected hybrid beamforming schemes based on the principle of alternating minimization. 
By contrast, a sub-connected architecture was proposed in \cite{XGao2016}, where the successive interference cancellation method is employed to derive closed-form digital and analog precoders.
Li \textit{et al.} \cite{XLi} introduced a deep learning approach to hybrid beamforming design with channel encoder and precoder networks to perform compressive channel sensing and hybrid beam prediction, respectively.
Moreover, a dynamic array-of-subarray architecture was further proposed, where the connections between analog precoders and RF chains can be dynamically adjusted to be adaptive to different channel conditions \cite{LYan2020}.
Although the analog part of hybrid beamforming can optimize the radiation pattern to some extent, it cannot completely eliminate the problem of pattern distortion.

\par To solve this problem, there have been some works showing that superdirectivity can be a potential approach to address the issue of distortion in edge radiation patterns. Specifically, by employing highly directive spatial beams, it is possible to achieve minor sidelobe energy leakage and improve the array gain.
In superdirectivity antenna theory, excitation current vector and array geometry design have been investigated in earlier research \cite{Uzkov,MTIvrlac2,E.Bjornson,A.Bloch,MTIvrlac}. It was demonstrated in \cite{Uzkov,MTIvrlac2} that the directivity of a linear array of $M$ isotropic antennas can reach $M^2$ as the spacing between antennas approaches zero.
Furthermore, the near-field characteristics of superdirectivity were investigated in \cite{Levin2021} from a linear transformation perspective. Subsequent research in this field \cite{E.Bjornson,A.Bloch,MTIvrlac} mainly focused on the excitation coefficient design for superdirective beams under different application limitations.  
Although several limitations of superdirective antenna arrays have been identified, such as low radiation efficiency, high excitation precision requirements, and low bandwidth, the mutual coupling effect for antenna elements with intervals less than half the wavelength is ignored in the aforementioned works. In 2019, Marzetta \textit{et al.} \cite{T.L.Marzetta} introduced the concept of superdirective arrays into communication systems. 
For further improved beamforming gain, subsequent works \cite{LHanSuperdirectiveConferene,LHanSuperdirectiveJournal,LHanSuperdirectiveMultiuser} proposed mutual coupling-based superdirective antenna arrays. Specifically, in the proposed analog beamforming technique, mutual coupling effects were utilized to generate superdirective spatial beams \cite{LHanSuperdirectiveConferene,LHanSuperdirectiveJournal}, which is a desirable radiation pattern for further improvement of system performance.
Through simulations and experimental measurements, $M^2$ directivity has been verified in the end-fire direction for linear arrays\cite{LHanSuperdirectiveJournal}. Although superdirectivity can address the pattern distortion issue for edge users, it may not provide coverage for other spatial regions.

\par Furthermore, superdirectivity requires densely placed antenna elements, which aligns with the recently proposed concept of holographic MIMO (HMIMO). HMIMO is an advanced technology building upon the concept of massive MIMO and involves integrating a massive (possibly infinite) number of antennas into a compact space. 
The pioneering works \cite{CHuang} and \cite{Emil2024} provided an initial introduction to the emerging HMIMO wireless communication, in particular the available HMIMO hardware architectures, functionalities and characteristics. Such surfaces comprising dense electromagnetic (EM) excited elements are capable of recording and manipulating impinging fields with utmost flexibility and precision, in addition to reduced cost and power consumption, thereby shaping arbitrary-intended EM waves with high energy efficiency.
Furthermore, a comprehensive overview of the latest advances in the HMIMO communications paradigm is presented in \cite{TGong1}, with a special focus on their physical aspects, theoretical foundations, as well as the enabling technologies for HMIMO systems.

\par In particular, the degree-of-freedom (DOF) characteristics of HMIMO were investigated in \cite{PizzoNyquist} for far-field communication scenarios, revealing the optimized spatial sampling locations of electromagnetic fields on HMIMO surfaces. To enable near-field HMIMO communications, the additional near-field DOF performance was revealed in \cite{RJi}. Moreover, \cite{ganxu2023, TGong2,LWeiTripolarNearField} further investigated near-field  characteristics and performances of HMIMO systems.
Additionally, the mutual coupling effect between antenna elements is prevalent in dense HMIMO arrays and the impact of mutual coupling on array radiation pattern and HMIMO performance was numerically studied in \cite{SYuanEffectsofMutualCoupling}. 
For the hardware design and implementation, an electromagnetic modulation scheme of HMIMO were proposed in \cite{Diboya1}, which is enabled by metamaterial antennas with tunable radiated amplitudes. Moreover, a holographic hybrid beamforming scheme was proposed in \cite{Diboya2,Diboya3} based on the aforementioned HMIMO hardware.


\subsection{Our Contributions}
In this paper, we introduce an electromagnetic hybrid beamforming (EHB) scheme based on three-dimensional (3D) antenna array architecture into HMIMO communications to implement a programmable array radiation pattern in any desired direction with minor radiation pattern distortion. Then, an electromagnetic channel model and corresponding beamforming algorithms are investigated to improve the spectral efficiency. The contributions of the paper are listed as follows.
  \begin{itemize}
    \item We first propose an EHB scheme with a holographic antenna array. Specifically, the EHB scheme consists of antenna excitation current vectors (analog beamforming) and digital precoding matrices. The implementation of analog beamforming enables the real-time adjustment of the radiation pattern to achieve an approximately constant gain in different intended directions and the digital beamforming is optimized based on the channel characteristics of analog beamforming. Moreover, we also develop electromagnetic signal and channel models in EHB communication schemes, which provide explanations of the modulation effects of the EHB scheme and the distribution of emitted energy from different perspectives.
    \item We investigate the hardware design of superdirectivity and 3D array geometry. Additionally, the optimal geometric arrangement of a 3D antenna array is also presented. Utilizing the 3D holographic antenna array structure and the customizable spatial radiation pattern, multiple superdirective beams in arbitrary directions can be generated by leveraging analog excitation vectors and taking the mutual coupling effects into account.
    \item Simulation results are finally provided to demonstrate that our proposed EHB scheme with a 3D holographic array achieves a relatively flat superdirective beamforming gain throughout the entire spatial domain. A sum rate gain of over 150\% compared to traditional beamforming algorithms is also obtained by our proposed EHB. 
  \end{itemize}

\subsection{Organization and Notation}
\par \textit{Organization:} The remainder of this paper is organized as follows. In Section \uppercase\expandafter{\romannumeral2}, we introduce the 
electromagnetic system model and the derivation of mutual coupling matrix. Then, an electromagnetic channel model is developed for further investigation. 
In Section \uppercase\expandafter{\romannumeral3}, we present the design principles of 3D antennas and the comprehensive EHB scheme, along with the corresponding algorithms. 
Numerical results are presented in Section \uppercase\expandafter{\romannumeral4} to evaluate the effectiveness of our proposed EHB scheme with 3D holographic array structure. Conclusions are drawn in Section \uppercase\expandafter{\romannumeral5}.

\par \textit{Notation:} Fonts $a$, $\boldsymbol{a}$ and $\boldsymbol{A}$ represent scalars, vectors and matrices, respectively. $\boldsymbol{A}^T, \boldsymbol{A}^H, \boldsymbol{A}^{-1}, \boldsymbol{A}^\dagger$ denote transpose, Hermitian, inverse and Moore-Penrose inverse operation respectively. $\|\boldsymbol{A}\|_F$ represents Frobenius norm of matrix $\boldsymbol{A}$. Operators $|\cdot|$ and $(\cdot)^*$ denote the modulus and conjugate operation respectively. Operators ${\rm Tr}(\cdot)$ and ${\rm diag}(\cdot)$ represent calculating the trace of a matrix and expanding a vector into a diagonal matrix, respectively.

\section{Electromagnetic System Model}
\par Different from the conventional hybrid beamforming scheme, we leverage the programmable electromagnetic characteristics of the holographic antenna array to solve the radiation pattern distortion. To this end, three challenges arise. \textit{First}, it is non-trivial to accurately model the array radiation pattern of a densely placed holographic antenna array due to the mutual coupling \cite{ChenXiaoming,Gradoni,Faqiri}. 
\textit{Second}, there is a need for further investigation into the impact of variations in the array radiation pattern on the performance of wireless communication systems.
\textit{Third}, although superdirective beams exhibit high directional gain, they can only be generated in end-fire directions for linear arrays \cite{LHanSuperdirectiveJournal,LHanSuperdirectiveMultiuser}, which limits its practical application. The solutions to these three challenges will be discussed in the following two sections.

\par In this section, we present the electromagnetic transmit signal and channel models in HMIMO communications.

\subsection{Electromagnetic Transmission Signal Model}
\vspace{0ex}
\par For densely placed holographic arrays, the effect of mutual coupling between the antennas can cause significant variations in the magnitude and phase characteristics of radiation patterns, resulting in deviations from the ideal analytical expressions. 
To accurately model the impact of these electromagnetic effects on signal transmission, we develop the following electromagnetic communication model.

\par Consider a HMIMO communication system depicted in Fig. \ref{transmitmodel}, where $N_T$ transmit antennas are densely arranged within a fixed array aperture.
The adaptability of this hardware structure can be achieved with two components: digital beamforming $\boldsymbol{W}_{BB}$ through RF chains and analog beamforming $\boldsymbol{i}$ through antenna excitation currents.
Firstly, we consider the conventional scenario where the antenna elements in the transmit MIMO array are assumed to be uncoupled. 
When the transmit signal $\boldsymbol{x}$ is emitted through the holographic array transmitter, it naturally undergoes modulation by the pattern function of each antenna element, which describes the power distribution of its spatial radiation. Therefore, we can model the emitted signal as:
\vspace{-0ex}
\begin{equation}
  \boldsymbol{x}^{(e)} = \boldsymbol{i} \odot \boldsymbol{g}(\boldsymbol{r,\theta,\phi}) \odot \boldsymbol{x} = \boldsymbol{i} \odot \boldsymbol{g}(\boldsymbol{r,\theta,\phi}) \odot \boldsymbol{W}_{BB}\boldsymbol{t},
  \label{uncoupled signal}
\end{equation}
where $\odot$ represents element-wise product;  
$\boldsymbol{t} \in \mathbb{C}^{K\times 1}$, $\boldsymbol{W}_{BB} \in \mathbb{C}^{N_T \times K}$,  $\boldsymbol{i} \in \mathbb{C}^{N_T \times 1}$ and $\boldsymbol{g}(\boldsymbol{r,\theta,\phi})$ represent the symbol vector for $K$ users, the digital beamforming matrix, the antenna excitation vector and a $N_T\times 1$ vector function representing the radiation pattern of antenna elements, respectively. The last one can be formulated as
\begin{equation}
  \begin{aligned}
    \boldsymbol{g}(\boldsymbol{r,\theta,\phi}) = [g(r_1,\theta_1,\phi_1),...&,g(r_n,\theta_n,\phi_n),... \\
    & ,g(r_{N_T},\theta_{N_T},\phi_{N_T})]^T,
    \label{radiation pattern function with r}
  \end{aligned}
\end{equation}
where $r_n$, $\theta_n$, and $\phi_n$ represent the distance, elevation angle, and azimuth angle, respectively, of the measurement point relative to the $n$th transmit antenna. Note that in the far-field communication scenario, the parameter $r_n$ can be disregarded, and $\theta_n$ and $\phi_n$ remain the same for all transmit antennas.

\par Secondly, we model the impact of the non-negligible mutual coupling effect on the transmission signal. 
After introducing the mutual coupling matrix \cite{LHanSuperdirectiveConferene}, the distorted transmission signal due to the mutual coupling effect can be written as
\begin{equation}
  \boldsymbol{x}^{(c)} = \boldsymbol{C}\boldsymbol{x}^{(e)} = \boldsymbol{C} \boldsymbol{i} \odot \boldsymbol{g}(\boldsymbol{r,\theta,\phi}) \odot \boldsymbol{W}_{BB}\boldsymbol{t},
  \label{coupled signal}
\end{equation}
where $c_{mn}$ represents the coherence of the $n$th element on the $m$th element. 

\vspace{-1ex}
\subsection{Radiation Pattern Model}
In this subsection, we investigate the mutual coupling matrix introduced in Section \uppercase\expandafter{\romannumeral2}. A in detail to model the actual radiation pattern of the densely placed holographic antenna array. Although some research has attempted to derive analytical expressions for mutual coupling between closely spaced antennas \cite{XChen,Masouros,PWang,S.Sadat}, practical results often exhibit significant deviations from these analytical expressions when the idealized conditions for theoretical derivations are not met. 
Therefore, in order to accurately evaluate the effects of mutual coupling on practical radiation patterns, we obtain the coupling matrix based on full-wave electromagnetic field simulation results in both coupled and uncoupled scenarios, as well as spherical wave expansions. This approach decomposes the electromagnetic field into a series of spherical wave coefficients using a set of orthogonal spherical wave bases. 

\par The actual radiation pattern of the $m$th element can be written as \cite{LHanSuperdirectiveJournal,S.Sadat}:
\begin{equation}
  g_{m,r}(r,\theta,\phi) = \sum\limits_{n=1} \limits^{N_T}c_{mn}i_{n}g_{n,i}(r,\theta,\phi),
\end{equation}
where the first subscript represents the antenna index, while the second subscript indicates either the realized radiation pattern ($r$) or the ideal radiation pattern ($i$).

\par The coupling matrix can be derived from spherical wave expansion, with the correlation matrix coefficients obtained by full-wave simulation results. 
With this expansion, the electric field $\boldsymbol{E}(r,\theta,\phi)$ radiated by one antenna element can be expanded as \cite{Cornelius}:
\vspace{-0.5ex}
\begin{equation}
\begin{aligned}
  \boldsymbol{e}(r,\theta,\phi)&=\kappa\sqrt{\eta}\sum\limits_{s=1}\limits^{2}\sum\limits_{n=1}\limits^{N}\sum\limits_{m=-n}\limits^{n} Q_{s,m,n}\boldsymbol{F}_{s,m,n}(r,\theta,\phi) \\
  & =\kappa\sqrt{\eta}\sum\limits_{s,m,n}Q_{s,m,n}\boldsymbol{F}_{s,m,n}(r,\theta,\phi),
  \label{sperical expansion}
\end{aligned}
\end{equation}
where $\kappa$ is the wavenumber and $\eta$ is the propagation medium intrinsic impedance. $Q_{s,m,n}$ is the spherical wave expansion coefficient and $s,m,n$ denotes electrical mode (TE/TM), degree of the wave and order of the wave, respectively. $N$ is the maximum order of available spherical wave functions, and is determined by the array aperture size. According to Chu-Harrington Limitation \cite{Chu,Harrington}, the available order of spherical wave expansion is determined by the electrical size of the transmitter as $N\approx\kappa a$, where $a$ is the minimum radius of a spherical surface which can enclose the antenna. To be specific, applying an order larger than $\kappa a$ will result in an exponentially increasing Q factor, bringing a narrow-banded antenna. Moreover, $\boldsymbol{F}_{s,m,n}(r,\theta,\phi)$ is the well-defined spherical wave function and the general form is shown in (\ref{spherical wave function}) \cite{Cornelius} (in the bottom of the next page). In this formula, $\boldsymbol{\hat{r}}$, $\boldsymbol{\hat{e}}_{\theta}$ and $\boldsymbol{\hat{e}}_{\phi}$ are unit vectors. $\boldsymbol{z}_{n}(\kappa r)$ is the spherical Hankel function of the first kind, and the term $P_n^{\left| m \right|}$(cos$\theta$) in (\ref{spherical wave function}) is the normalized associated Legendre function. For TM waves, the magnetic field is described by $\boldsymbol{F}_{1mn}$, and $\boldsymbol{F}_{2mn}$ represents the corresponding electric fields. For TE waves, this is interchanged.

The electric field can be decomposed into:
\begin{equation}
  \boldsymbol{e}(r,\theta,\phi)=\left[E^{\hat{r}}(r,\theta,\phi),E^{\hat{\theta}}(r,\theta,\phi),E^{\hat{\phi}}(r,\theta,\phi)\right],
  \vspace{-0.5ex}
\end{equation}
with the superscripts $\hat{r}$, $\hat{\theta}$ and $\hat{\phi}$ representing $r$-component, $\theta$-component, and $\phi$-component, respectively.

\begin{figure*}[htbp]
  \centering
  \includegraphics[width=5.5in]{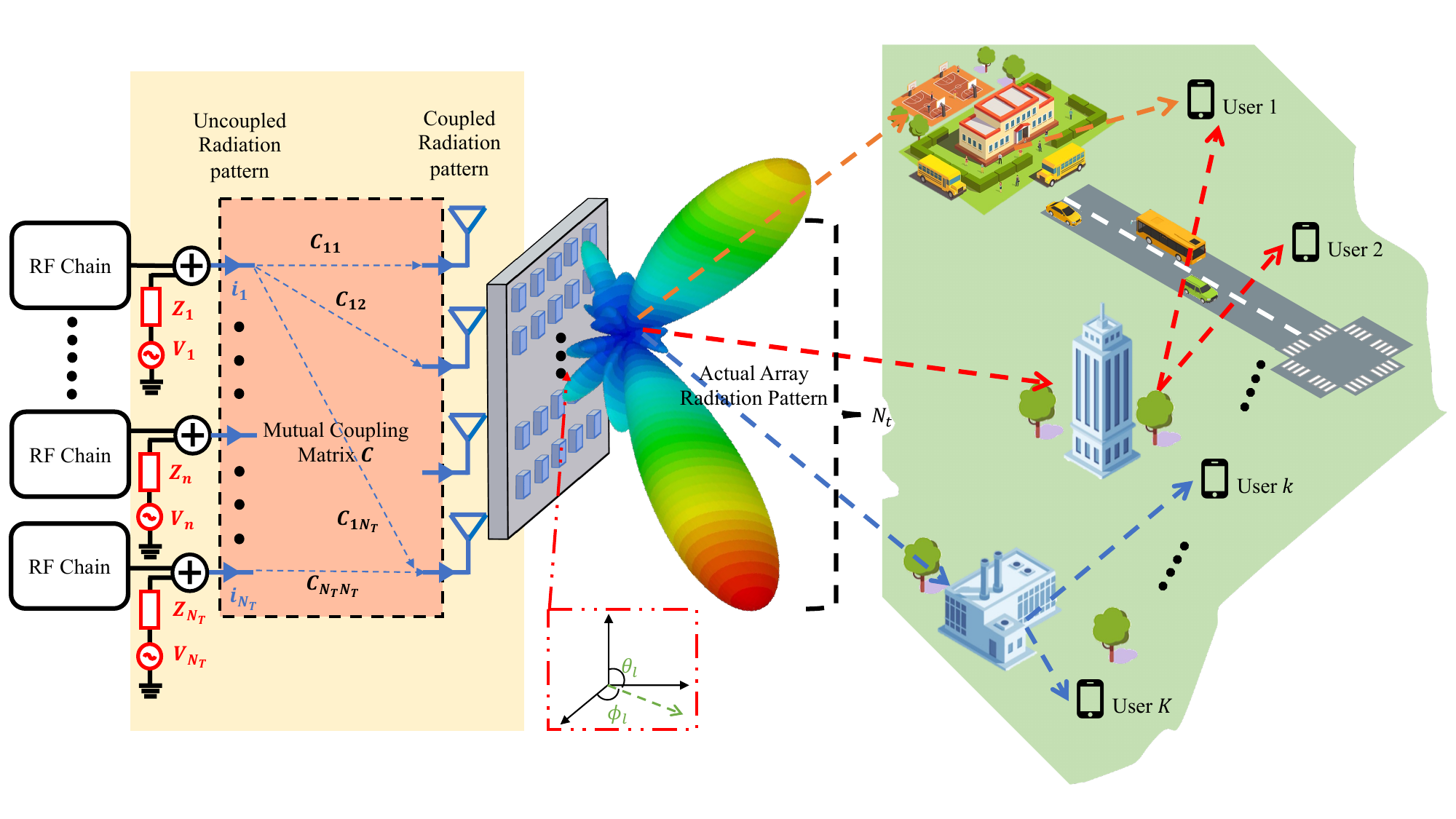}
  \vspace{-2ex}
  \caption{Electromagnetic Hybrid Beamforming Architecture.}
  \label{transmitmodel}
  \vspace{-3ex}
\end{figure*}

\begin{figure*}[!b]
  \begin{subequations}
    \begin{align}
      &\boldsymbol{F}_{mn}(r,\theta,\phi)=\frac{1}{\sqrt{2\pi}}\frac{1}{\sqrt{n(n+1)}}\left(-\frac{m}{\left|m \right|}\right)^m \boldsymbol{z}_n(\kappa r)P_n^{\left|m\right|}(\cos\theta)e^{-jm\phi} \label{spherical wave function 1}, \\ 
      &\boldsymbol{F}_{1mn}(r,\theta,\phi)=\nabla\boldsymbol{F}_{mn}(r,\theta,\phi)\times \boldsymbol{\hat{r}} \notag \\
      &=\frac{1}{\sqrt{2\pi}}\frac{1}{\sqrt{n(n+1)}}\left(-\frac{m}{\left|m\right|}\right)^m\left\{ \boldsymbol{z}_n(\kappa r)\frac{imP_n^{\left|m\right|}(\cos\theta)}{\sin\theta}e^{-jm\phi}\boldsymbol{\hat{e}}_{\theta} - \boldsymbol{z}_{n}(\kappa r)\frac{-jmP_n^{\left|m\right|}(\cos\theta)}{\sin\theta}e^{-jm\phi}\boldsymbol{\hat{e}}_{\phi}\right\}, \\
      &\boldsymbol{F}_{2mn}=\kappa^{-1}\nabla\times\boldsymbol{F}_{1mn}(r,\theta,\phi) .
    \end{align}
    \label{spherical wave function}
  \end{subequations}
\end{figure*}

\par The spherical wave expansion can be interpreted as the decomposition of an electromagnetic field into a series of orthogonal components with different basic modes. To determine the spherical mode coefficients $Q_{s,m,n}$, we can reformulate the expansion (\ref{sperical expansion}) into matrix form
\begin{equation}
  \boldsymbol{E}=\kappa\sqrt{\eta}\boldsymbol{F}\boldsymbol{Q},
\end{equation}
where $\boldsymbol{E}\in \mathbb{C}^{3P\times N_T}$, $\boldsymbol{F}\in \mathbb{C}^{3P\times 2N(N+2)}$, and $\boldsymbol{Q} \in \mathbb{C}^{2N(N+2) \times N_T}$represent the spatially sampled electric field matrix, the basic spherical mode matrix, and the spherical mode coefficient matrix, respectively, with $P$ being the number of spatial sampling points.
Their specific representations for the $n$th antenna are as follows \cite{KBelmkaddem}
\begin{equation}
  \resizebox{0.9\hsize}{!}{$
  \begin{aligned}
    &\boldsymbol{E}(:,n)=[E^{\hat{r}}_n(r_1,\theta_1,\phi_1),E^{\hat{\theta}}_n(r_1,\theta_1,\phi_1),E^{\hat{\phi}}_n(r_1,\theta_1,\phi_1), \\ 
    & \cdots, E^{\hat{r}}_n(r_P,\theta_P,\phi_P),E^{\hat{\theta}}_n(r_P,\theta_P,\phi_P),E^{\hat{\phi}}_n(r_P,\theta_P,\phi_P)]^T,
  \end{aligned}$}
\end{equation}
\begin{equation}
  \begin{aligned}
    \boldsymbol{Q}(:,n)=[Q_{1,-1,1}^n,Q_{2,-1,1}^n,\cdots,Q_{2,N,N}^n]^T,
    \label{spherical mode coefficient}
  \end{aligned}
\end{equation}
\begin{equation}
  \resizebox{0.85\hsize}{!}{$
  \boldsymbol{F}=
  \left[
    \begin{matrix}
      F^{\hat{r}}_{1,1,1}(r_1,\theta_1,\phi_1) & \cdots & F^{\hat{r}}_{2,N,N}(r_1,\theta_1,\phi_1) \\
      F^{\hat{\theta}}_{1,1,1}(r_1,\theta_1,\phi_1) & \cdots & F^{\hat{\theta}}_{2,N,N}(r_1,\theta_1,\phi_1) \\
      F^{\hat{\phi}}_{1,1,1}(r_1,\theta_1,\phi_1) & \cdots & F^{\hat{\phi}}_{2,N,N}(r_1,\theta_1,\phi_1) \\
      \vdots & \vdots & \vdots \\
      F^{\hat{r}}_{1,1,1}(r_P,\theta_P,\phi_P) & \cdots & F^{\hat{r}}_{2,N,N}(r_P,\theta_P,\phi_P)\\ 
      F^{\hat{\theta}}_{1,1,1}(r_P,\theta_P,\phi_P) & \cdots & F^{\hat{\theta}}_{2,N,N}(r_P,\theta_P,\phi_P)\\
      F^{\hat{\phi}}_{1,1,1}(r_P,\theta_P,\phi_P) & \cdots & F^{\hat{\phi}}_{2,N,N}(r_P,\theta_P,\phi_P)\\ 
    \end{matrix}
  \right],$}
\end{equation}
where $F^{\hat{r}}_{s,m,n}(r_n,\theta_n,\phi_n)$, $F^{\hat{\theta}}_{s,m,n}(r_n,\theta_n,\phi_n)$ and $F^{\hat{\phi}}_{s,m,n}(r_n,\theta_n,\phi_n)$ represent different orthogonal components of the spherical wave function. 
Based on the above matrix representation, we can derive the spherical mode coefficients in (\ref{spherical mode coefficient}) as 
\begin{equation}
  \boldsymbol{Q}=\frac{1}{\kappa \sqrt{\eta}}\left(\boldsymbol{F}\right)^{\dagger}\boldsymbol{E}.
\end{equation}

\par After obtaining the full-wave simulation electromagnetic fields for both the coupled and uncoupled cases, the mutual coupling matrix corresponding to the holographic antenna array can be represented by 
\begin{equation}
  \boldsymbol{C} = \boldsymbol{Q}_C^T(\boldsymbol{Q}_S^T)^\dag,
\end{equation}
where $\boldsymbol{Q}_C$ is the spherical mode coefficient matrix for coupled cases and $\boldsymbol{Q}_S$ is the uncoupled one. In summary, we investigate the effects of mutual coupling by regarding the actual radiation pattern of each antenna element as a superposition of the ideal radiation patterns of all elements in the array. The mutual coupling coefficients serve as the weighting coefficients for this superposition. Note that the mutual coupling coefficients are computed based on radiation fields, and therefore the computed coefficients are at the field level.

\par The impact of mutual coupling is illustrated in Fig. \ref{spatial beam}. The first two subfigures depict the spatial radiation patterns of an isolated antenna element and a coupled one.
Fig. \ref{spatial beam}(c) displays the spatial superdirective beam achieved by leveraging mutual coupling, whereas Fig. \ref{spatial beam}(d) shows the spatial beam using a conventional beamforming algorithm. It is evident that neglecting mutual coupling leads to a significant reduction in directivity.

  \begin{figure}[!ht]
    \vspace{-2ex}
    \centering
    \includegraphics[width=3.3in]{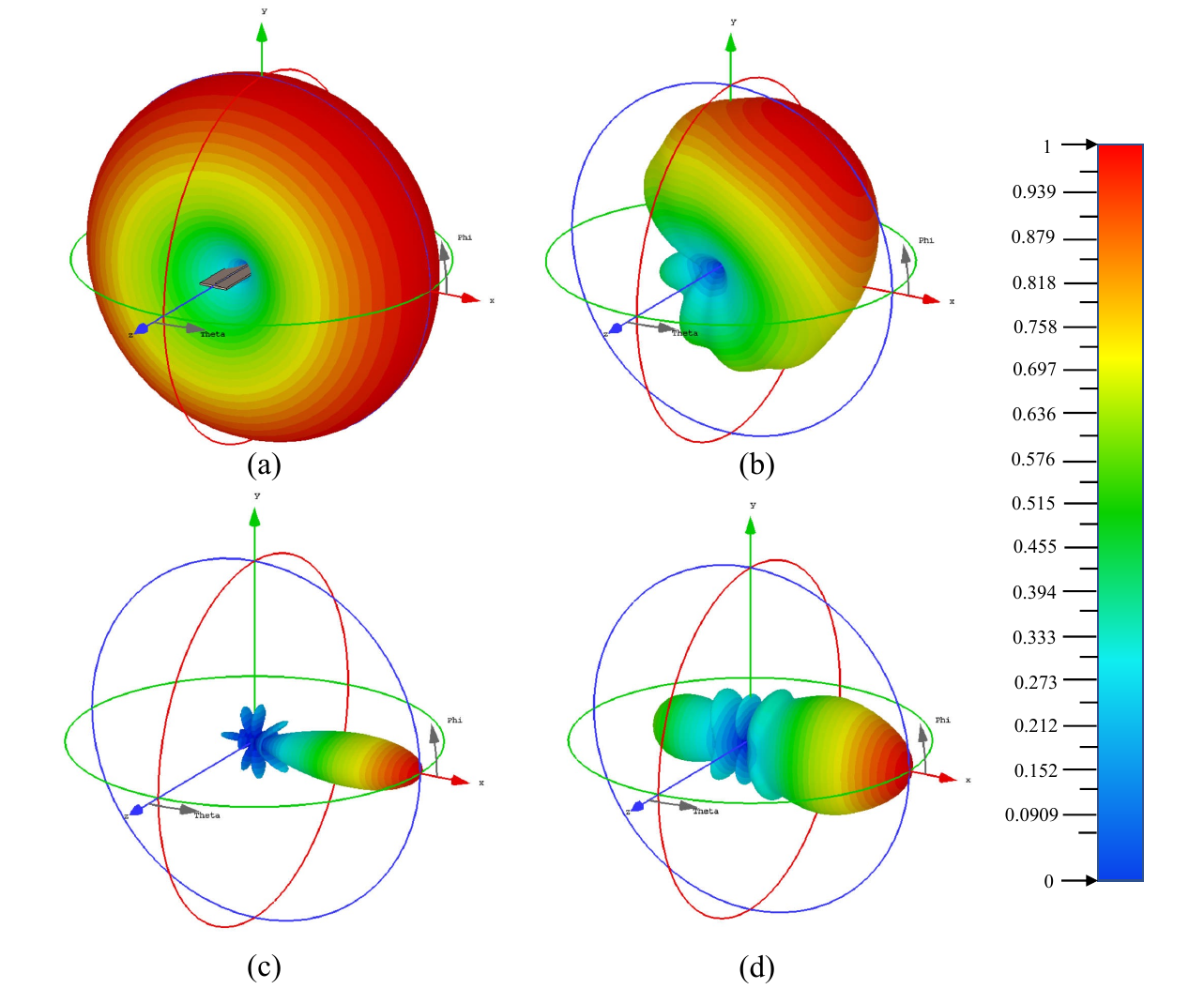}
    \caption{Radiation pattern for (a) An isolated antenna; (b) One antenna in holographic antenna array with 0.2$\lambda$ spacing; (c) Spatial beam based on mutual coupling; (d) Spatial beam without considering mutual coupling.}
    \label{spatial beam}
    \vspace{-1ex}
  \end{figure}

\subsection{Electromagnetic Channel Model}
\par In this subsection, we incorporate the above electromagnetic characteristics of holographic antenna array into a conventional geometrically based stochastic channel model (GSCM) to take both the stochastic characteristics of wireless environment and the spatial radiation pattern of the transmitter into consideration. 

\par Assuming that there are $L$ scatters and $K$ users in the far field region of interest (i.e., we omit variable $\boldsymbol{r}$ in the radiation pattern function (\ref{radiation pattern function with r})), the classical Saleh-Valenzuela channel model for a single-antenna user $k$ can be written as \cite{DTse,OmarEl2014}
\begin{equation}
  \resizebox{0.89\hsize}{!}{$
  \boldsymbol{h}_{k} = \sum \limits_{l=1}^{L}\sqrt{P_{lk}}\alpha_{lk}  \exp(-j2\pi f\tau_{lk}) \boldsymbol{C} \boldsymbol{i} \odot \boldsymbol{g}(\boldsymbol{\theta,\phi}) \delta(\Omega - \Omega_l),$}
  \label{channel vector for k}
\end{equation}
where $\tau_{lk}$, $\alpha_{lk}$ and $P_{lk}$ represent the propagation time delay, channel coefficient and power of the $l$th path, respectively. $\Omega_l$ represents the solid angle of the $l$th scatterer relative to the transmitter.
With $s_{lk} = \sqrt{P_{lk}}\alpha_{lk}\exp(-j2\pi f \tau_{lk})$, the received signal for user $k$ can be formulated as
\begin{equation}
  \begin{aligned}
  y_k &= \int \boldsymbol{h}_k^T \boldsymbol{x} d\Omega = \int \sum \limits_{l=1}^{L}s_{lk}\delta(\Omega - \Omega_l) \mathbf{1}^T\boldsymbol{x}^{(c)} d\Omega + n\\
  & = \sum \limits_{l=1}^{L} s_{lk}\mathbf{1}^T (\boldsymbol{C} \boldsymbol{i} \odot \boldsymbol{g}(\Omega_l) \odot \boldsymbol{x}) + n\\
  & = \sum \limits_{l=1}^{L} s_{lk}\boldsymbol{x}^T (\boldsymbol{C} \boldsymbol{\hat{I}} \boldsymbol{g}(\Omega_l)) + n,
  \end{aligned}
  \label{receivedsignalk}
\end{equation}
where $\mathbf{1} \in \mathbb{C}^{N_T \times 1}$, $\boldsymbol{\hat{I}} = {\rm diag}(\boldsymbol{i})$, and $n$ represent an all-one vector, a diagonal matrix representing analog antenna excitations, and the Gaussian noise term, respectively. Based on (\ref{receivedsignalk}), the received signal can be written in the matrix form
\begin{equation}
  \boldsymbol{y} = (\boldsymbol{C} \boldsymbol{\hat{I}} \boldsymbol{G S})^T \boldsymbol{x} + \boldsymbol{n},
  \label{receivedsignal}
\end{equation}
where $\boldsymbol{S} \in \mathbb{C}^{L\times K}$, $\boldsymbol{G} \in \mathbb{C}^{N_T\times L}$, and $\boldsymbol{n}$ are the channel coefficient matrix, the spatial sample of the radiation pattern, and an i.i.d. Gaussian noise vector with element variance $\sigma^2$, respectively. 
$\boldsymbol{G}$ describes the spatial radiation pattern corresponding to the $L$ scatterers and can be obtained after acquiring the spatial coordinates of scatterers in the environment through sensing or channel estimation.
In practical systems, channel state information at the receiver can be obtained via training and subsequently shared with the transmitter via feedback. Techniques for efficient channel estimation leveraging the geometric nature of wireless environments can be found in literature \cite{Bajwa2010,Rangan2011,chaokaiwen2015,ganxu2022,additional1,additional2}, 
as well as deep learning-based approaches \cite{yuwenyang2019,hengtaohe2018,ruxin2020,chunchangjae2019,fenghaozhu2024} or beam scanning techniques \cite{HurSooyoung2013,yacongding2018,Emil2023}.
The derived electromagnetic GSCM can be written as 
\begin{equation}
  \boldsymbol{H} = (\boldsymbol{C}\boldsymbol{\hat{I}}\boldsymbol{G}\boldsymbol{S})^T \in \mathbb{C}^{K\times N_T},
  \label{Hmatrix}
\end{equation}

\vspace{-3ex}
\subsection{Holographic Antenna Array Radiation Efficiency}
\par Besides the impact of distorted radiation patterns on system performance, another significant factor affecting the realized communication rate is antenna radiation efficiency. Based on existing research on superdirective antenna arrays \cite{Hansen,McLean}, it has been observed that as the interval space between antenna elements decreases, the achievable directivity increases. However, corresponding radiation efficiency decreases, resulting in a potential deterioration in performance due to the decrease in the actual realized antenna gain. Therefore, in practical terms, a denser antenna array does not necessarily guarantee better system performance.
The definition of antenna efficiency is given as follows
\begin{equation}
  e_{\text{rad}}=\frac{P_{\text{rad}}}{P_{\text{acc}}}=\frac{G(\theta,\phi)}{D(\theta,\phi)},
\end{equation}
where $e_{\text{rad}}$, $P_{\text{rad}}$, and $P_{\text{acc}}$ are radiation efficiency, the radiated power, and the accepted power of the antenna, respectively. $D(\theta,\phi)$ and $G(\theta,\phi)$ are defined as the directivity and gain of the antenna, and their specific definitions are provided as follows
\begin{equation}
  \begin{aligned}
    D(\theta,\phi)&=4\pi\cdot\frac{U}{P_{\text{rad}}}, \\
    G(\theta,\phi)&=4\pi\cdot\frac{U}{P_{\text{rad}}}\frac{R_{\text{rad}}}{R_{\text{rad}}+R_{\text{dis}}},
  \end{aligned}
\end{equation}
where $R_{\text{rad}}$, $R_{\text{dis}}$, and $U = r^2W_{\text{rad}}$ represent the radiation resistance, dissipation resistance, and radiation intensity ( Watt/unit solid angle), respectively. 

\par Since we achieve superdirectivity by leveraging the mutual coupling effect of the holographic array transmitter, a smaller element spacing is preferred for a stronger coupling and more directive spatial beams to increase the concentration of radiated energy. On the other hand, holographic antenna arrays with smaller space between elements tend to have larger dissipation resistances, resulting in a lower radiation efficiency. Therefore, there exists a tradeoff between the radiation efficiency and the directive gain since the actual realized array gain is the product of them (i.e., $G(\theta,\phi) = e_{\text{rad}} \cdot D(\theta,\phi)$), which means that there is an optimal array geometry, and our following simulation results also further validate this conjecture.
\vspace{-2ex}
\section{Proposed EHB Scheme}
To leverage superdirectivity and programmable array radiation pattern, the EHB scheme is introduced in this section. 
We consider the antenna radiation efficiency to indicate the received signal
power in practical scenario. Then, the 3D holographic array architecture and  the electromagnetic beamforming algorithm will be provided in the sequel. 

\subsection{3D Holographic Antenna Array Structure}
The directivity expression for an antenna array can be written as follows in the absence of mutual coupling:
\begin{equation}
  D = \frac{\boldsymbol{i}^H \left|g(\theta,\phi)\right|^2 \boldsymbol{e}(\theta,\phi)\boldsymbol{e}^H(\theta,\phi)\boldsymbol{i}}{\boldsymbol{i}^H \left(\int_{0}^{2\pi}\int_{0}^{\pi}\left|g(\theta,\phi)\right|^2 \boldsymbol{e}(\theta,\phi)\boldsymbol{e}^H(\theta,\phi)\sin(\theta)d\theta d\phi\right)\boldsymbol{i}},
\end{equation}
where $\boldsymbol{e} = \left[e^{-j\kappa\boldsymbol{\hat{r}} \cdot \boldsymbol{r}_1},e^{-j\kappa\boldsymbol{\hat{r}} \cdot \boldsymbol{r}_2},\hdots,e^{-j\kappa\boldsymbol{\hat{r}} \cdot \boldsymbol{r}_{N_T}}\right]^T$ represents the array steering vector, with $\boldsymbol{\boldsymbol{\hat{r}}}$ being the unit vector in direction $(\theta,\phi)$ and $\boldsymbol{r}_{n}$ being the Cartesian coordinates of the $n$th antenna. To maximize the achievable directivity, the optimal $\boldsymbol{i}$ can be derived as:
\begin{equation}
  \boldsymbol{i}_{\text{opt}} = \gamma \boldsymbol{Z}^{-1}\boldsymbol{e}(\theta,\phi),
  \label{supercurrent}
\end{equation}
with $\boldsymbol{Z} = \frac{\int_{0}^{2\pi}\int_{0}^{\pi}\left|g(\theta,\phi)\right|^2 \boldsymbol{e}(\theta,\phi)\boldsymbol{e}^H(\theta,\phi)\sin(\theta)d\theta d\phi}{\int_{0}^{2\pi}\int_{0}^{\pi}\left|g(\theta,\phi)\right|^2\sin(\theta)d\theta d\phi}$ denoted as the impedance matrix and $\gamma$ denoting the power normalization factor. For lossless antennas, the radiation power $P_{\text{rad}}$ equals the input power $P_{\text{in}}$, and the maximum array gain can be written as follows \cite{MTIvrlac}:
\begin{equation}
  \begin{aligned}
    A = \max \frac{P_{\text{Rx}}}{P_{\text{Rx}}|_{N_T=1}} = N\frac{\boldsymbol{e}^H(\theta,\phi) \boldsymbol{Z}^{-1}\boldsymbol{e}(\theta,\phi)}{\boldsymbol{e}^H(\theta,\phi)\boldsymbol{e}(\theta,\phi)},
  \end{aligned}
\end{equation}
where $P_{\text{Rx}}$ represents the received power. Specifically, the array gain depends on the the beamforming direction and array geometry (via $\boldsymbol{Z}$). For isotropic elements along the $z$ axis, the largest directivity is obtained in the direction $\theta = 0$ (end-fire), and approaches $N_T^2$, for small antenna separation.

\par Inspired by the former conclusion, we find that the superdirectivity origins from the dense element arrangements along a specific direction and the corresponding optimal excitation currents. Thus, the $\boldsymbol{Z}$ matrix of a densely placed 3D holographic antenna array would fulfill the geometry requirements of superdiectivity in arbitrary directions. Considering the mutual coupling effect of holographic antenna array on actual radiation pattern, the optimal excitation currents in (\ref{supercurrent}) can be modulated as:
\begin{equation}
  \boldsymbol{i}^{(c)}_{\text{opt}} = \gamma \boldsymbol{C}^{-1}\boldsymbol{Z}^{-1}\boldsymbol{e}(\theta,\phi),
\end{equation}
and multiple superdiective beams can be implemented through a linear combination of excitation current vectors in different directions as $\boldsymbol{i}^{(c)}_{\text{opt}} = \boldsymbol{i}^{(c)}_{\text{opt},1} + \boldsymbol{i}^{(c)}_{\text{opt},2}$. 
As an example, Fig. \ref{antenna structure} shows a typical 3D geometry with multiple layers of planar aperture. Moreover, two superdirective spatial beams in orthogonal directions (broadside and end-fire directions) are realized, demonstrating the effectiveness of our proposed structure and design principle. Therefore, the real-time optimization technique of circuit parameters can be utilized to provide an approximately coincident array gain in any desired direction and eliminate the edge pattern distortion.

 
\begin{figure}[!t]
  \centering
  \includegraphics[width=3.3in]{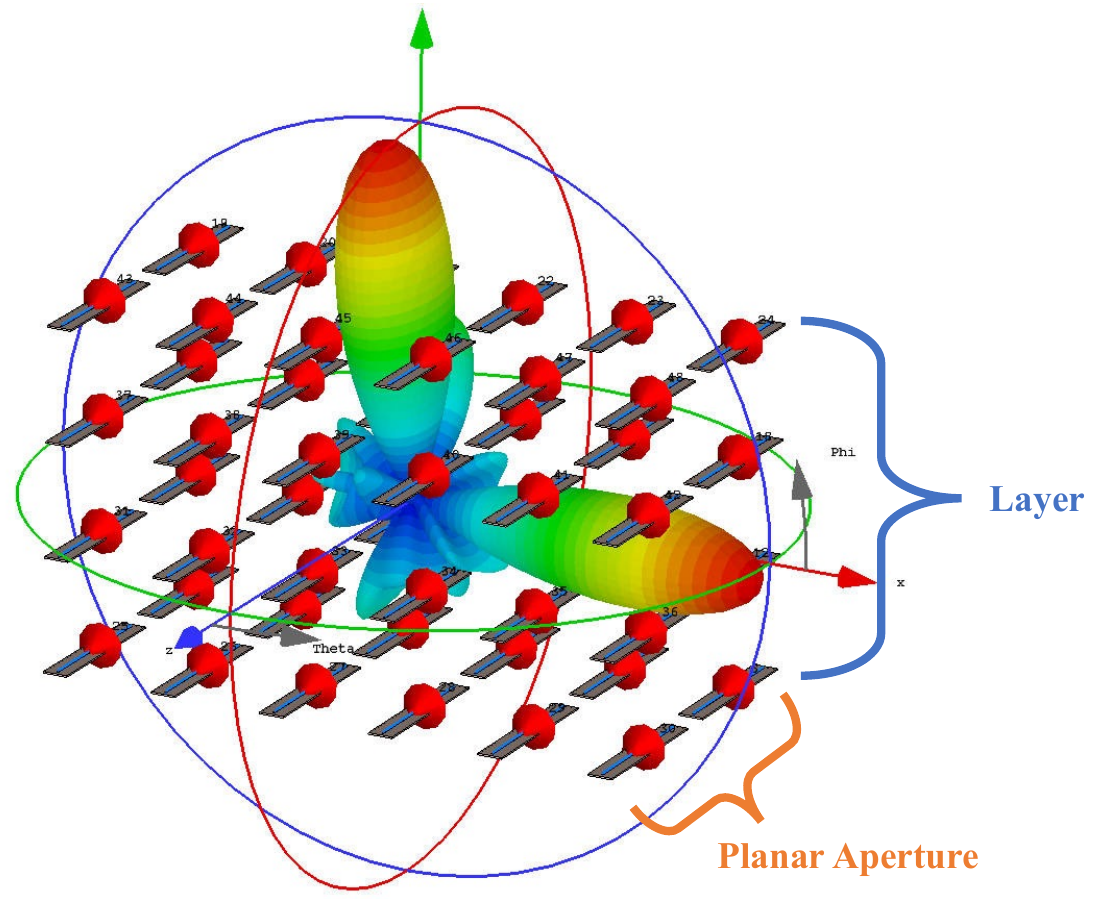}
  \caption{Multiple superdirective beams of 3D holographic antenna array.}
  \label{antenna structure}
  \vspace{-1ex}
\end{figure}

\vspace{-1ex}
\subsection{Proposed EHB Algorithm}
Based on the electromagnetic system model and the 3D antenna array structure that we have introduced, we aim to fully leverage the programmable electromagnetic characteristics of the holographic antenna array. To achieve this, we introduce EHB to modulate the radiation pattern of the array and mitigate radiation distortion.
Specifically, the characteristics of our proposed EHB algorithm are twofold:
i) Electromagnetic analog precoding method: By real-time programming the circuit parameters, including parameters of programmable antennas and the coupling between antennas, load impedance or excitation currents of the transmitter, this approach can realize the real-time adjustment of the radiation pattern to achieve an approximately coincide gain for users in various directions.
ii) Electromagnetic digital precoding method: Its objective is to obtain the optimal match filter to further improve the spectral efficiency of communication system based on the channel information after the analog precoding. Then, analog precoding and digital precoding are iteratively optimized until they both achieve the optimum performance.

\par Next we show the design details of the analog excitation vector $\boldsymbol{i}$ and the digital precoding matrix $\boldsymbol{W}_{BB}$ to optimize overall system performance based on the electromagnetic GSCM derived in Section \uppercase\expandafter{\romannumeral2}. 

\subsubsection{Single-User Case}
First we consider single-user communication scenario. Referring to the general formulas (\ref{receivedsignal}), the received signal in a single-user communication scheme can be formulated as 

\begin{equation}
  \vspace{-1ex}
  \begin{aligned}
    y &= \sqrt{e_{\text{rad}}}\boldsymbol{h}^T\boldsymbol{x}+n \\
    &= \sqrt{e_{\text{rad}}}(\boldsymbol{C}\boldsymbol{\hat{I}}\boldsymbol{Gs})^T\boldsymbol{x} + n \\
    &= \sqrt{e_{\text{rad}}}(\boldsymbol{C}\boldsymbol{\hat{I}}\boldsymbol{Gs})^T\boldsymbol{ w_{BB}}t + n ,
  \end{aligned}
\end{equation}
and the corresponding received signal-to-noise-ratio (SNR) can be represented as 
\begin{equation}
  \vspace{-1ex}
  \begin{aligned}
    \gamma &= \frac{e_{\text{rad}}| (\boldsymbol{C}\boldsymbol{\hat{I}}\boldsymbol{Gs})^T\boldsymbol{w}_{BB}|^2}{\sigma^2},
  \end{aligned}
\end{equation}
where $\boldsymbol{w}_{BB} \in \mathbb{C}^{N_T \times 1}$ represents the digital beamforming vector and $n$ is the noise term following Gaussian random distribution $\mathcal{N}(0,\sigma^2)$. 
The transmit symbol is normalized such that $|t|^2=1$. Then the corresponding EHB optimization problem can be formulated as
\begin{equation}
  \begin{aligned}
    P1: &\max \limits_{\boldsymbol{\hat{I}},\boldsymbol{w}_{BB}} \gamma \\
    s.t. & \qquad \| \boldsymbol{\hat{I}} \|_F^2 \leq P_{T,\text{analog}},  \|\boldsymbol{w}_{BB}\|^2_2 \leq P_{T,\text{digital}}\\
    & \qquad \boldsymbol{\hat{I}} = {\rm diag}(\boldsymbol{i}) \\
    & \qquad P_{T,\text{analog}} + P_{T,\text{digital}} \leq P_T ,
  \end{aligned}
  \label{P1}
\end{equation}
where $P_T$, $P_{T,\text{analog}}$, and $P_{T,\text{digital}}$ represent the available transmit power, the analog excitation power, and the digital precoding power, respectively \cite{linglongdai2023}.
Additionally, the radiation efficiency, which is mainly determined by the geometry of the 3D holographic array, also impacts the received SNR. However, for a fixed structure, different beamforming directions are not expected to cause a significant variation in radiation efficiency and the mutual coupling matrix.

\par Therefore, the proposed EHB algorithm for the single-user scheme consists of two main procedures. Firstly, a search is performed on the value of the element spacing to find an optimized 3D antenna array under the given operation frequency. This optimization aims to strike a balance between beamforming gain and radiation efficiency.
In practice, this can be achieved by referring to the performance curves of realized gain as shown in the simulation section. 
Secondly, the optimal solutions of analog and digital precoding vectors are applied to maximize the received SNR.

\par Initially, we apply the following transformation:
\begin{equation}
  \begin{aligned}
    \boldsymbol{\hat{I}}\boldsymbol{Gs} = {\rm diag}(\boldsymbol{Gs})\boldsymbol{i},
  \end{aligned}
  \label{transform1}
\end{equation}
by substituting the transformation (\ref{transform1}) into (\ref{P1}), we obtain the following formulation:
\begin{equation}
  \begin{aligned}
    P2: &\max \limits_{\boldsymbol{i},\boldsymbol{w}_{BB}} |\boldsymbol{i}^T {\rm diag}(\boldsymbol{Gs}) \boldsymbol{C}^T\boldsymbol{w}_{BB}|^2 \\
    s.t. & \qquad \| \boldsymbol{i} \|_2^2 \leq P_{T,\text{analog}},  \|\boldsymbol{w}_{BB}\|^2_2 \leq P_{T,\text{digital}}\\
    & \qquad P_{T,\text{analog}} + P_{T,\text{digital}} \leq P_T ,
  \end{aligned}
  \label{P2}
\end{equation}
\par By expanding the objective function into conjugate form, we have: 
\begin{equation}
  \begin{aligned}
    |\boldsymbol{i}^T {\rm diag}(\boldsymbol{Gs}) \boldsymbol{C}^T \boldsymbol{w}_{BB}|^2 = &\boldsymbol{i}^T {\rm diag}(\boldsymbol{Gs}) \boldsymbol{C}^T\boldsymbol{w}_{BB} \\
    & \boldsymbol{w}_{BB}^H \boldsymbol{C}^* {\rm diag}(\boldsymbol{Gs})^*\boldsymbol{i}^*,
  \end{aligned}
\end{equation}
which is a quadratic form function for variable $\boldsymbol{i}$, and the optimal solution $\boldsymbol{i}_{\text{opt}}$ corresponds to the eigenvectors of the intermediate matrix with respect to the variable $\boldsymbol{i}$:
\begin{equation}
  \boldsymbol{i}_{\text{opt}} = \alpha {\rm diag}(\boldsymbol{Gs})^*\boldsymbol{C}^H\boldsymbol{w}_{BB}^*,
  \label{singlei}
\end{equation}
where $\alpha$ is the power normalization factor to satisfy the analog power constraint. Similarly, we can derive the optimal digital precoding vector as 
\begin{equation}
  \boldsymbol{w}_{BB,\text{opt}} = \beta \boldsymbol{C}^*{\rm diag}(\boldsymbol{Gs})^*\boldsymbol{i}^*,
  \label{singlew}
\end{equation}
where $\beta$ is the power normalization factor of digital precoding. 
By combining (\ref{singlei}) and (\ref{singlew}), it is easy to prove that 
\begin{equation}
  \begin{aligned}
    & \frac{1}{\alpha\beta^*}\boldsymbol{i}_{\text{opt}} = {\rm diag}(\boldsymbol{Gs})^*\boldsymbol{C}^H \boldsymbol{C}{\rm diag}(\boldsymbol{Gs})\boldsymbol{i}_{\text{opt}}\\
    & \frac{1}{\alpha^*\beta}\boldsymbol{w}_{BB,\text{opt}} = \boldsymbol{C}^*{\rm diag}(\boldsymbol{Gs})^* {\rm diag}(\boldsymbol{Gs})\boldsymbol{C}^T\boldsymbol{w}_{BB,\text{opt}},
  \end{aligned}
  \label{eigende}
\end{equation}
by substituting (\ref{eigende}), (\ref{singlei}) and (\ref{singlew}) into $P2$, it can be easily proven that under constant power constraints, the optimal $\boldsymbol{i}$ and $\boldsymbol{w}_{BB}$ are the eigenvectors corresponding to the largest eigenvalues of matrices ${\rm diag}(\boldsymbol{Gs})^*\boldsymbol{C}^H \boldsymbol{C}{\rm diag}(\boldsymbol{Gs})$ and $\boldsymbol{C}^*{\rm diag}(\boldsymbol{Gs})^* {\rm diag}(\boldsymbol{Gs})\boldsymbol{C}^T$, respectively, and they can be obtained by taking the SVD decomposition.

\subsubsection{Multi-User Scheme}
In this subsection, we extend the previous single-user communication scheme to a multi-user scenario. Based on the electromagnetic transmit signal and channel models introduced in section \uppercase\expandafter{\romannumeral2}, the received signal can be expressed as follows:
\begin{equation}
  \begin{aligned}
    \boldsymbol{y} &= \sqrt{e_{\text{rad}}}\boldsymbol{H}\boldsymbol{x}+ \boldsymbol{n}\\
    &= \sqrt{e_{\text{rad}}}(\boldsymbol{C}\boldsymbol{\hat{I}}\boldsymbol{GS})^{T}\boldsymbol{x}+\boldsymbol{n} \\
    &= \sqrt{e_{\text{rad}}}(\boldsymbol{C}\boldsymbol{\hat{I}}\boldsymbol{GS})^{T}\boldsymbol{W}_{BB}\boldsymbol{t}+\boldsymbol{n},
    \label{received signal}
  \end{aligned}
\end{equation}
From (\ref{received signal}), it can be seen that radiation efficiency, mutual coupling matrix $\boldsymbol{C}$, and analog beamforming excitation $\boldsymbol{\hat{I}}$ are the main factors influencing system performance.  
Then, we can formulate the sum rate maximization problem as
\begin{equation}
  \begin{aligned}
    P3 : &\max \limits_{\boldsymbol{\hat{I}},\boldsymbol{W}_{BB}} \sum \limits_{k=1}^{K} \log(1+\gamma_k)\\
    s.t. & \qquad \| \boldsymbol{\hat{I}} \|_F^2 \leq P_{T,\text{analog}},  \|\boldsymbol{W}_{BB}\|^2_F \leq P_{T,\text{digital}}\\
    & \qquad \boldsymbol{\hat{I}} = {\rm diag}(\boldsymbol{i}) \\
    & \qquad P_{T,\text{analog}} + P_{T,\text{digital}} \leq P_T ,
  \end{aligned}
  \label{sumratemax}
\end{equation}
and the received signal-to-noise-plus-interference ratio (SINR) for user $k$ can be written as follows due to the mutual interference between different users
\begin{equation}
  \gamma_k = \frac{ e_{\text{rad}} |[(\boldsymbol{C}\boldsymbol{\hat{I}}\boldsymbol{GS})(:,k)]^T \boldsymbol{W}_{BB}(:,k)|^2 }{ e_{\text{rad}}\sum \limits_{i\neq k}|[(\boldsymbol{C}\boldsymbol{\hat{I}}\boldsymbol{GS})(:,k)]^T\boldsymbol{W}_{BB}(:,i)|^2 + \sigma^2},
  \label{multiuserSNR}
\end{equation}
where $\boldsymbol{A}(:,k)$ represents the $k$th column of matrix $\boldsymbol{A}$.

\par It can be seen that Eq. (\ref{multiuserSNR}) is a classical multi-user digital precoding problem. 
However, for analog precoding design, it is still challenging to directly optimize the precoding matrix $\boldsymbol{I}$ based on (\ref{multiuserSNR}). Therefore, we reformulate (\ref{multiuserSNR}) and propose a unified optimization procedure for both analog and digital precodings. By defining $\boldsymbol{X} = \boldsymbol{G}\boldsymbol{S} \in \mathbb{C}^{N_T \times K}$, we can rewrite the $k$th column of the matrix $\boldsymbol{C}\boldsymbol{\hat{I}}\boldsymbol{GS}$ as:

\vspace{-2ex}
\begin{equation}
  \begin{aligned}
    (\boldsymbol{C}\boldsymbol{\hat{I}}\boldsymbol{GS})(:,k) &= 
    \boldsymbol{C} \times 
    \left[\begin{matrix}
          i_{1} &  0 & \hdots\\ 
          0 & \ddots & \vdots\\ 
          \vdots & \hdots & i_{N_T} \\
    \end{matrix}\right] \times 
    \left[\begin{matrix}
      x_{1k} \\ 
      x_{2k} \\ 
      \vdots \\
      x_{N_Tk}
    \end{matrix}\right] \\
    & = \boldsymbol{C} \times 
    \left[\begin{matrix}
          x_{1k} &  0 & \hdots\\ 
          0 & \ddots & \vdots\\ 
          \vdots & \hdots & x_{N_Tk} \\
    \end{matrix}\right] \times 
    \left[\begin{matrix}
      i_{1k} \\ 
      i_{2k} \\ 
      \vdots \\
      i_{N_Tk}
    \end{matrix}\right] \\
    & = \boldsymbol{C}{\rm diag}(\boldsymbol{X}(:,k))\boldsymbol{i},
  \end{aligned}
  \label{outi}
\end{equation}
where $x_{nk}$ represents the $(n,k)$th element of matrix $\boldsymbol{X}$. Leveraging (\ref{outi}), the SINR of user $k$ can be written as 
\begin{equation}
  \begin{aligned}
    \gamma_k  &= \frac{e_{\text{rad}}|\boldsymbol{i}^T {\rm diag}(\boldsymbol{X}(:,k)) \boldsymbol{C}^T \boldsymbol{W}_{BB}(:,k)|^2}{e_{\text{rad}}\sum \limits_{i\neq k} |\boldsymbol{i}^T {\rm diag}(\boldsymbol{X}(:,k)) \boldsymbol{C}^T \boldsymbol{W}_{BB}(:,i)|^2 + \sigma^2}
  \end{aligned}
\end{equation}
Then, the sum rate maximization problem for multi-user EHB can be reformulated as 
\begin{equation}
  \begin{aligned}
    P4 : & \max \limits_{\boldsymbol{i},\boldsymbol{W}_{BB}} \sum \limits_{k=1}^{K} \log(1+\gamma_k)\\
    s.t. & \qquad \| \boldsymbol{i} \|_2^2 \leq P_{T,\text{analog}},  \|\boldsymbol{W}_{BB}\|^2_F \leq P_{T,\text{digital}}\\
    & \qquad P_{T,\text{analog}} + P_{T,\text{digital}} \leq P_T ,
    \label{P4}
  \end{aligned}
\end{equation}
which is a non-convex optimization problem and it is also difficult to directly solve it. 
To address this, we employ the AO framework that iteratively optimizes the electromagnetic analog and digital precodings. Specifically, an SDR technique
is employed to solve the auxiliary matrix $\boldsymbol{R}$ of analog excitations $\boldsymbol{i}$, while auxiliary matrices $\boldsymbol{F}_k$ are used to optimize the precoding matrices (digital part) \cite{WHaoRobustDesign}\cite{LuoZhiQuan2010}.
The solution details of EHB is outlined as follows:

\par Firstly, by introducing auxiliary variables $t_k$ and $u_k$, we can reformulate $P4$ to optimize the electromagnetic digital beamforming when fixing the analog part
\begin{subequations}
  \vspace{-0.5ex}
    \begin{equation}
    P5: \max \limits_{\boldsymbol{F}_k} \sum \limits_{k=1}^{K} \log(1+t_k) \\
    \end{equation}
    \begin{equation}
    s.t. \qquad u_k \geq \sum \limits_{i\neq k} {\rm Tr}(\boldsymbol{H}_k\boldsymbol{F}_i) + \sigma^2 \\
    \end{equation}
    \begin{equation}
    \qquad \frac{t_k^{(n)}}{2u_k^{(n)}}u^2_k + \frac{u^{(n)}_k}{2t^{(n)}_k}t^2_k \leq {\rm Tr}(\boldsymbol{H}_k\boldsymbol{F}_k) \\
    \end{equation}
    \begin{equation}
    \qquad \sum \limits_{k=1}^{K} {\rm Tr}(\boldsymbol{F}_k) \leq P_{T,\text{digital}}\\
    \end{equation}
    \begin{equation}
    \qquad {\rm rank}(\boldsymbol{F}_k) = 1, \boldsymbol{F}_k \succeq \boldsymbol{0},
    \end{equation}
    \label{P5}
\end{subequations}
\vspace{-3ex}
  
\noindent where the superscript $(n)$ is the iteration index,  $\boldsymbol{F}_k = \boldsymbol{W}_{BB}(:,k)\boldsymbol{W}_{BB}^{H}(:,k)$ and $\boldsymbol{H}_k = \boldsymbol{C}^*{\rm diag}(\boldsymbol{X}^*(:,k))\boldsymbol{i}^*\boldsymbol{i}^T {\rm diag}(\boldsymbol{X}(:,k))\boldsymbol{C}^T$. 

\par Secondly, by fixing the digital beamforming and introducing similar auxiliary variables, we can reformulate $P4$ similarly as above to optimize the electromagnetic analog beamforming as follows:
\begin{subequations}
  \vspace{-0.5ex}
  \begin{equation}
    P6: \max \limits_{\boldsymbol{G}} \sum \limits_{k=1}^K \log(1+t_k) \\
    \end{equation}
    \begin{equation}
    s.t. \qquad u_k \geq \sum \limits_{i \neq k} {\rm Tr}(\boldsymbol{R}\boldsymbol{Q}_k) + \sigma^2 \\
    \end{equation}
    \begin{equation}
    \qquad \frac{t_k^{(n)}}{2u_k^{(n)}}u_k^2 + \frac{u_k^{(n)}}{2t_k^{(n)}}t_k^2 \leq {\rm Tr}(\boldsymbol{R}\boldsymbol{Q}_k) \\
    \end{equation}
    \begin{equation}
    \qquad {\rm Tr}(\boldsymbol{R}) \leq P_{T,\text{analog}} \\
    \end{equation}
    \begin{equation}
    \qquad {\rm rank}(\boldsymbol{R})=1, \boldsymbol{R} \succeq \boldsymbol{0},
    \end{equation}
    \label{P6}
\end{subequations}
\vspace{-3ex}

\noindent where $\boldsymbol{R} = \boldsymbol{i}^*\boldsymbol{i}^T$ and $\boldsymbol{Q}_i = {\rm diag}(\boldsymbol{X}(:,k))\boldsymbol{C}^T\boldsymbol{W}_{BB}(:,i)\boldsymbol{W}_{BB}^H(:,i)\boldsymbol{C}^*{\rm diag}(\boldsymbol{X}^*(:,k))$.

  \par Finally, after dropping the non-convex rank-one constraint, the two optimization problems mentioned above both can be reformulated as SDR problems, and solved alternatively by optimizing the digital precoding matrices and analog excitation vectors using convex solvers such as CVX. Note that the SDR problem is a convex optimization problem, and each iteration's solution is the optimal. Therefore, iteratively solving the optimization problems and updating the variables will either increase or at least maintain the value of the objective function. Given limited transmit power, the proposed iterative algorithm guarantees that the obtained sum rate value forms a monotonically non-decreasing sequence with an upper bound.
  The convergence condition for Algorithm \ref{EHB algorithm} is achieved in terms of the rate value in two successive iterations with variations less than a specified threshold, denoted as $\epsilon$. 

  \par Now we reconsider the rank-one constraint. As mentioned above, to formulate SDR problems, we removed the rank-one constraint in the above two optimization procedures. Fortunately, it can be proven that the solutions for $P5$ and $P6$ also satisfy the rank-one constraint.

\begin{theorem}
  The obtained auxiliary digital precoding matrix $\boldsymbol{F}_k$ and auxiliary analog precoding matrix $\boldsymbol{G}$ in $P5$ and $P6$ satisfy {\rm rank}($\boldsymbol{F}_k$) $=1$, {\rm rank}($\boldsymbol{G}$) $=1$.
\end{theorem}

\begin{proof}
  Please see Appendix A.
\end{proof}

\begin{algorithm}
  \caption{EHB Scheme for 3D Holographic Array}
  \label{EHB algorithm}
  \begin{algorithmic}[1]
  \Require {Number of antennas $N_T$, user number $K$, frequency $f$, antenna element radiation pattern $g(r,\theta,\phi)$}, geometrical channel coefficients $s_{lk}$, scatter angular information $\Omega_l$ and mutual coupling matrix $\boldsymbol{C}$
  \Ensure Electromagnetic analog precoding vector $\boldsymbol{i}$ and digital precoding matrix $\boldsymbol{W}_{BB}$
      \For{$l = 1,2,...,L$} 
        \State Compute the EM channel vector $\boldsymbol{h}_k$ by (\ref{channel vector for k});
      \EndFor
      \State Construct the EM channel matrix $\boldsymbol{H} = \left[\boldsymbol{h}_1,\hdots,\boldsymbol{h}_K\right]^T$;
      \State Initialize $\boldsymbol{i} = \mu_i \boldsymbol{C}^{-1}\boldsymbol{Z}^{-1}\boldsymbol{e}$ and $\boldsymbol{W}_{BB} = \mu_w\frac{\boldsymbol{H}_{w}^H}{(\boldsymbol{H}_{w}\boldsymbol{H}_{w}^H)^{-1}}$, with $\boldsymbol{H}_{w} = \left[ \boldsymbol{C} {\rm diag}(\boldsymbol{i}) \boldsymbol{GS}\right]^T$;
      \For{$i = 1,2,...,N_{iter}$}
        \State{Determine the auxiliary rank-one matrices $\boldsymbol{F}_k$ for 
        \Statex \quad \, digital precoding matrix optimization based on (\ref{P5})};
        \State {Determine the auxiliary rank-one matrix $\boldsymbol{R}$ for antenna
        \Statex \quad \, excitation vector optimization based on (\ref{P6})};
      \EndFor
      \For{$k = 1,2,...,K$}
        \State Compute $\boldsymbol{f}_k$ through SVD of matrix $\boldsymbol{F}_k$;
      \EndFor
      \State {Construct $\boldsymbol{W}_{BB} = \left[\boldsymbol{f}_1,\hdots,\boldsymbol{f}_K\right]$};
      \State {Obtain $\boldsymbol{i}$} through SVD of $\boldsymbol{R}$;
      \State \Return {$\boldsymbol{i}$ and $\boldsymbol{W}_{BB}$}
  \end{algorithmic}
\end{algorithm}

\par The employed EHB algorithm is summarized in \textbf{Algorithm 1}. We first compute the electromagnetic channel vector $\boldsymbol{h}_k$ for each user and construct the overall channel matrix $\boldsymbol{H}$ based on the estimated GSCM, as shown in lines 1-4. 
We initialize the analog excitation vector and digital beamforming matrix based on the impedance matrix $\boldsymbol{Z}$ and steering vector $\boldsymbol{e}$. The initial point is determined by adopting the excitation currents that maximize the array directivity in the absence of mutual coupling, with $\boldsymbol{Z}$ defined in (\ref{supercurrent}), as shown in line 5.
Then, according to the AO optimization based on SDR method in $P5$ and $P6$, we can obtain the optimized rank-one auxiliary matrices, as shown in lines 7-8. By taking the SVD of auxiliary matrices, we can verify the correctness of SDR solutions by checking the rank and obtain the corresponding beamforming vectors, as shown in lines 11 and 14. Finally, the digital beamforming matrix can be obtained by concatenating the derived precoders for each user, as shown in line 13. 

\vspace{-2ex}
\subsection{Complexity Analysis}
In this subsection, we analyze the computational complexity of our proposed EHB scheme. Recall that the above convex restrictions in $P5$ and $P6$ involve only linear matrix inequality (LMI) and second-order cone (SOC) constraints. Therefore, they can all be solved by a standard interior-point method (IPM) \cite{IPM1,IPM2,IPM3}. This suggests that the worst-case runtime of such a method can be used to compare the computational complexities of the formulations. From \cite{IPMcomplexity}, the basic elements in the complexity analysis of IPMs can be formulated as follows
\vspace{-0.5ex}
\begin{subequations}
  \begin{align}
    &\min \limits_{\boldsymbol{z}\in \mathbb{R}^\mathcal{N}} \boldsymbol{c}^T\boldsymbol{z} \\
    &\sum \limits_{i=1}^{\mathcal{N}}z_i \boldsymbol{A}^j_i - \boldsymbol{B}^j \in \mathbb{S}_{+}^{k_j} \qquad \mathrm{for} \ j = 1,\hdots,q,\label{complexity2}\\ 
    & \boldsymbol{D}^j \boldsymbol{z} - \boldsymbol{b}^j \in \mathbb{L}^{k_j} \qquad \ j=q+1,\hdots, p,
  \end{align}
  \label{complexity}
  \vspace{-2ex}
\end{subequations}

\noindent where $\mathcal{N}$ is the number of unknowns, $\boldsymbol{A}^j_i, \boldsymbol{B}^j \in \mathbb{S}_{+}^{k_j}$ for $i=1,\hdots,\mathcal{N}$ and $j=1,\hdots,q$; $\boldsymbol{D}^j \in \mathbb{R}^{k_j\times \mathcal{N}}$ and $\boldsymbol{b}^j \in \mathbb{R}^{k_j}$ for $j=q+1,\hdots,p$; $\boldsymbol{c} \in \mathbb{R}^\mathcal{N}$; $\mathbb{S}_{+}^{k}$ is the set of $k \times k$ real positive semidefinite matrices and $\mathbb{L}^{k}$ is the second-order cone of dimension $k \geq 1$. Specifically, linear constraint $\boldsymbol{a}^T\boldsymbol{z} - b \geq 0$ is equivalent to size 1 LMI constraint $\boldsymbol{a}^T\boldsymbol{z} - b \in \mathbb{S}^1_+$ and can be put into form (\ref{complexity2}). 
\par The complexity of solving (\ref{complexity}) through a generic IPM can be decomposed into the following two parts \cite{KYWang}:

\par \textit{Iteration Complexity}: Let $\epsilon > 0 $ be the iteration accuracy, the number of iterations required to reach an $\epsilon$-optimal solution of (\ref{complexity}) is on the order of $\sqrt{\beta}\cdot \ln(1/\epsilon)$, where $\beta(\Lambda) = \sum \limits_{j=1}^{q}k_{j} + 2(p-q)$ is  the barrier parameter associated with the cone $\Lambda=\prod_{j=1}^{q}\mathbb{S}_{+}^{k_j}\times \prod_{j=q+1}^p\mathbb{L}^{k_j}$. Roughly speaking, the barrier parameter measures the geometric complexity of the conic constraints in (\ref{complexity}).
\par \textit{Per-iteration Computation Complexity}: In each iteration, the computation cost is determined by (i) the formation of the $\mathcal{N}\times \mathcal{N}$ coefficient matrix $\boldsymbol{H}$ of the linear system, and (ii) the factorization of $\boldsymbol{H}$. The cost of forming $\boldsymbol{H}$ is on the order of 
  \begin{equation}
    \vspace{-0.5ex}
    \begin{aligned}
      &C_{form} = \mathcal{N}\sum_{j=1}^{q}k_j^3 + \mathcal{N}^2\sum_{j=1}^{q}k_j^2 + \mathcal{N}\sum_{j=q+1}^p k_j^2  \\
      &C_{fact} = \mathcal{N}^3,
    \end{aligned}
  \end{equation}
\par The total computation cost per iteration is on the order of $C_{form} + C_{fact}$, and the corresponding complexity of a generic IPM is on the order of $\sqrt{\beta(\Lambda)} \cdot (C_{form}+C_{fact})\dot \ln(1/\epsilon)$.

\par For P5, there are $K+1$ LMI constraints of size 1, $K$ LMI constraints of size $N_T$ and $K$ second-order cone constraints of size 3. For P6, there are $K+1$ LMI constraints of size 1, 1 LMI constraints of size $N_T$ and $K$ second-order cone constraints of size 3. Their corresponding complexity orders are summarized in Table. \ref{complexity order table}. 
\begin{table}[h!]
  \renewcommand\arraystretch{1.7} 
  \begin{center}
    \caption{COMPLEXITY ANALYSIS OF P5 AND P6}
    \label{complexity order table}
    \begin{tabular}{m{1cm}<{\centering}|m{6cm}<{\centering} } 
      \hline 
      \textbf{Method} & {\makecell[c]{\textbf{Complexity Order} \\ \textbf{( $n_1$ = $\mathcal{O}$ $(K N_T^2$)) ( $n_2$ = $\mathcal{O}$ $(N_T^2$))}}}  \\
      \hline 
      \multirow{2}{*}{P5} & $\sqrt{KN_T+3K+1}\cdot n_1 \cdot$ \\ 
      & $\left[ (KN_T^3+10K+1)+n_1(KN_T^2+K+1)+n_1^2 \right]$\\
      \hline
      \multirow{2}{*}{P6} & $\sqrt{N_T+3K+1}\cdot n_2 \cdot$ \\
      & $\left[ (N_T^3+10K+1) + n_2(N_T^2+K+1)+n_2^2 \right]$ \\
      \hline
    \end{tabular}
  \end{center}
  \vspace{-5ex}
\end{table}

\section{Numerical Results}
In this section, simulation results are presented to evaluate the performance of our proposed EHB communication scheme based on 3D holographic antenna arrays.
The communication system is assumed to operate at 1.6 GHz.
A 3D antenna array consisting of printed dipole antennas is equipped at BS to serve $K$ users, with specific antenna hardware parameters as follows: The dipole antenna is printed on a FR-4 substrate measuring $12.2 \enspace \text{mm} \times 78 \enspace \text{mm}$, and the dimension of a dipole element is $1 \enspace \text{mm} \times 71.5 \enspace \text{mm}$.
The digital transmit power constraint is determined based on the SNR settings in simulations while the analog power constraint for antenna excitation is set to 1 W.
We consider the area corresponding to four linear half-wavelength spaced antennas as the aperture area of a single layer. Therefore, the number of antennas of one layer is denoted as $N_t = \lfloor \frac{1.5\lambda}{d}\rfloor+1$, and the total number of antennas is given by $N_T=\nu N_t$, where $d$ and $\nu$ represent the element spacing and number of 3D layers, respectively.
The convergence parameter  for Algorithm \ref{EHB algorithm} is set to $\epsilon=0.1$. Additionally, we assume that there are $L$ scatterers and $K$ users distributed randomly in the azimuth plane, considering a non-line-of-sight environment.

\par For the comparison benchmark, we use the current excitation vector corresponding to maximum directivity in the absence of mutual coupling $\boldsymbol{i} = \mu_i \boldsymbol{Z}^{-1}\boldsymbol{e}$ in \cite{LHanSuperdirectiveJournal} and the corresponding zero-forcing (ZF) digital precoding matrix. 
In the subsequent simulations, we compare the performance of our proposed EHB scheme against the conventional beamforming schemes. 1000 sets of channel coefficients are sampled and the numerical performance results are averaged over these samples.

\begin{figure}[!ht]
  \centering
  \includegraphics[width=3.3in]{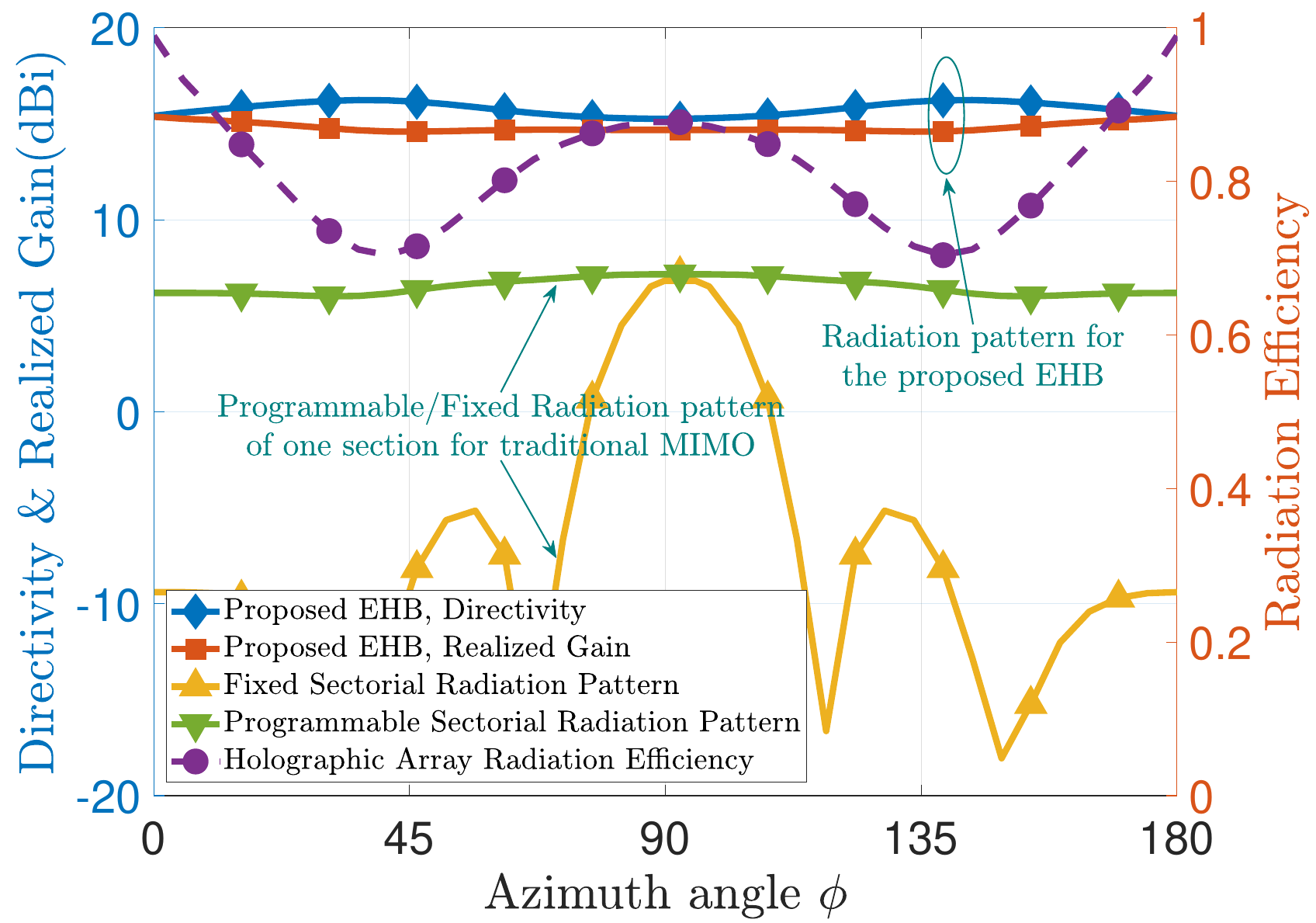}
  \caption{Achievable directivity and realized gain for EHB system.}
  \label{DirectivityvsAngle}
  \vspace{-1ex}
\end{figure}

\par Fig. \ref{DirectivityvsAngle} shows the achieved directivity and realized gain of the 3D holographic antenna array for different transmit azimuth angles. From the figure, we can see that our proposed 3D array achieves a relatively smooth gain over the whole space, which proves the effectiveness of its structure.
Additionally, the directivity curves exhibit peaks at 45 and 135 degrees. This phenomenon can be attributed to the fact that these angles are in the vicinity of the diagonal direction of the 3D antenna array. That is to say, from the perspective of beamforming direction, the diagonal direction corresponds to the longest quasi-linear array within the array structure. The presence of these peaks in the directivity curves validates our intuitive design principle.

\begin{figure}[!ht]
  \centering
  \includegraphics[width=3.3in]{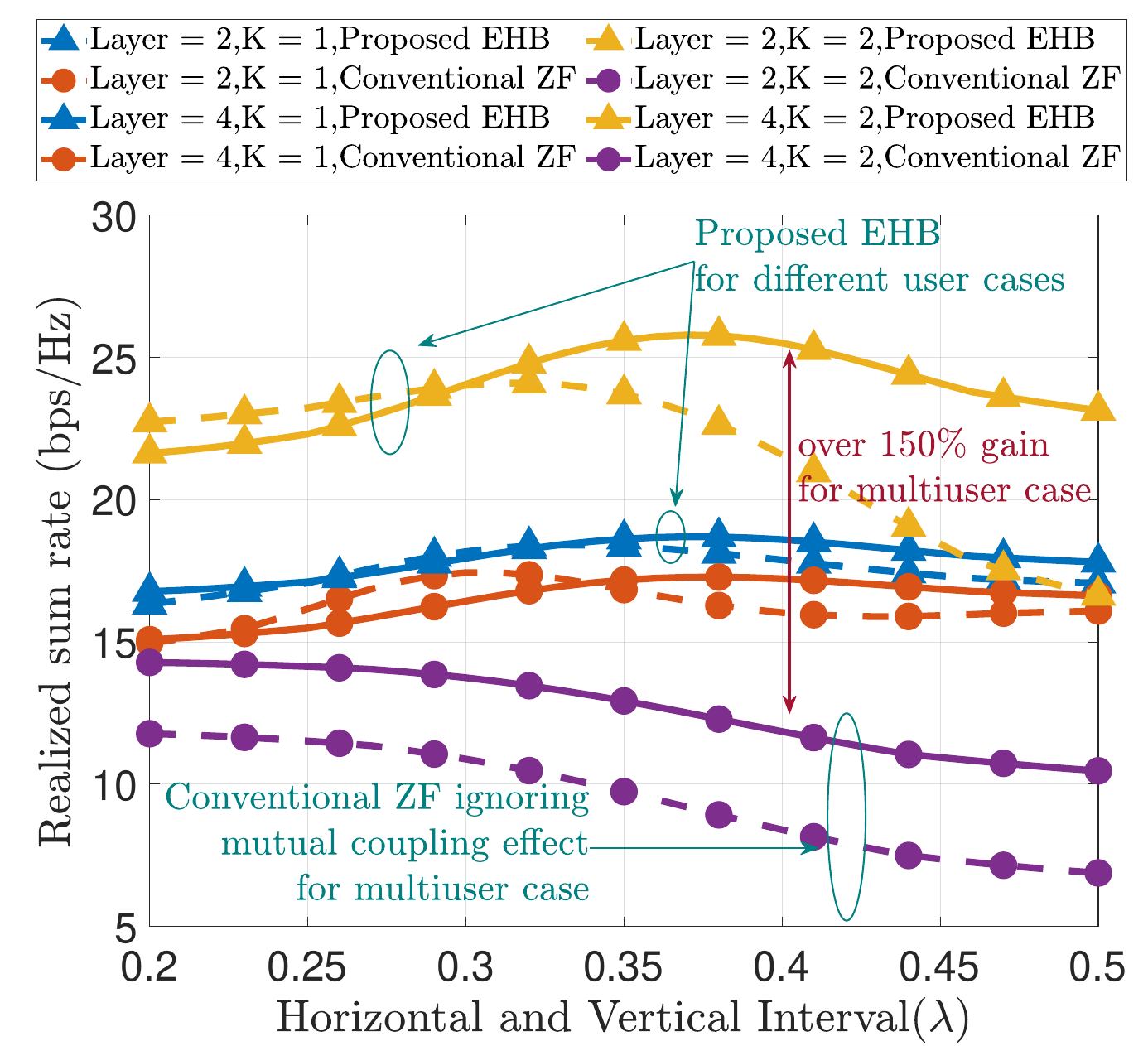}
  \caption{Achievable rate for EHB system with SNR = $20$ dB.}
  \label{SingleMultiUserRate}
\end{figure}

\par Fig. \ref{SingleMultiUserRate} shows the rate performance of the EHB scheme under different element spacing and user settings. 
As can be seen, the ZF precoding without considering the mutual coupling effect achieves the worse rate performance than that of the proposed scheme. Specifically, for the single-user case, the gain of EHB is relatively limited around 10\% on average at a 20 dB SNR level.
However, in multi-user communication scenarios, exploiting programmable radiation patterns can significantly reduce inter-user interference and improve beam focusing. By leveraging the spatial superdirective beams, the proposed scheme achieves a substantial capacity improvement with an average performance gain of over 150\%. Note that the optimal element spacing varies for different 3D antenna structures. This indicates that the geometrical configuration of the 3D holographic antenna array plays a crucial role in achieving the best system performance.
The results clearly demonstrate that our proposed scheme outperforms the conventional precoding algorithms in HMIMO communications.

\par To evaluate the performance of the EHB scheme for different transmission power cases, Fig. \ref{SumRatevsSNR} shows the sum rate performance under different SNR settings when the antenna spacing is fixed at 0.35 wavelength. Simulation results indicate that larger performance gains are observed as the SNR increases. For comparison, we introduce full-digital precoding and traditional hybrid beamforming in \cite{hybridcompare1,hybridcompare2}.
Specifically, the 4-layers optimal sum rate gain is achieved by the multi-user 3D antenna array, demonstrating the effectiveness of our proposed EHB algorithm and architecture. The conventional ZF curve tends to saturate because of the mutual interference induced by the mutual coupling effect.

\begin{figure}[!h]
  \centering
  \includegraphics[width=3.3in]{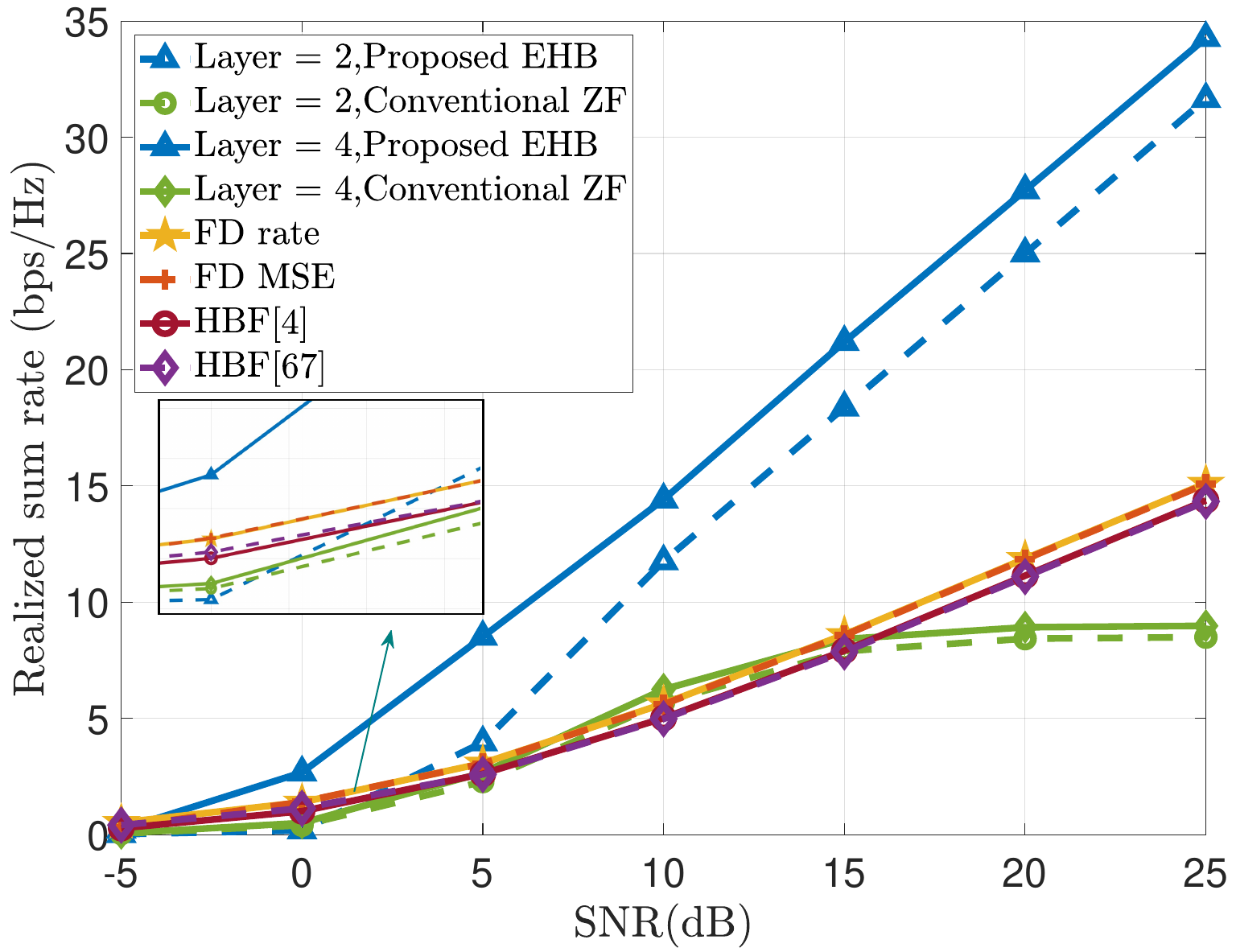}
  \caption{Achievable sum rate for multi-user EHB system with fixed antenna aperture.}
  \label{SumRatevsSNR}
  \vspace{0ex}
\end{figure}

\par Fig. \ref{DirectivityvsDistance} shows the achievable directivity and corresponding radiation efficiency of the 3D holographic antenna array for different antenna spacings.
For comparison, we denote theoretical gain limitation as the realized gain benchmark, expressed as $G = \frac{4\pi A}{\lambda^2}$ \cite{Hannan}, where $A$ refers to the array aperture area and $\lambda$ represents the wavelength. 
It can be observed that our proposed EHB scheme achieves superdirectivity based on the 3D array architecture while adhering to a fixed aperture area constraint.
It is also shown that the radiation efficiency of the 3D array experiences a significant reduction as the antenna spacing decreases, indicating stronger mutual coupling effects.
Besides, for a holographic array transmitter with a fixed aperture, when the antenna spacing is smaller than a threshold, for example, $0.3\lambda$ in this case, the mutual coupling effect is severe and requires highly accurate excitation current control, which is not achievable under practical hardware constraints, leading to a decreased directivity. The aforementioned optimal antenna spacing in an HMIMO array can vary depending on the specific 3D structure of the array.
The array radiation efficiency is obtained through CST simulations. 

\begin{figure}[!ht]
  \centering
  \includegraphics[width=3.3in]{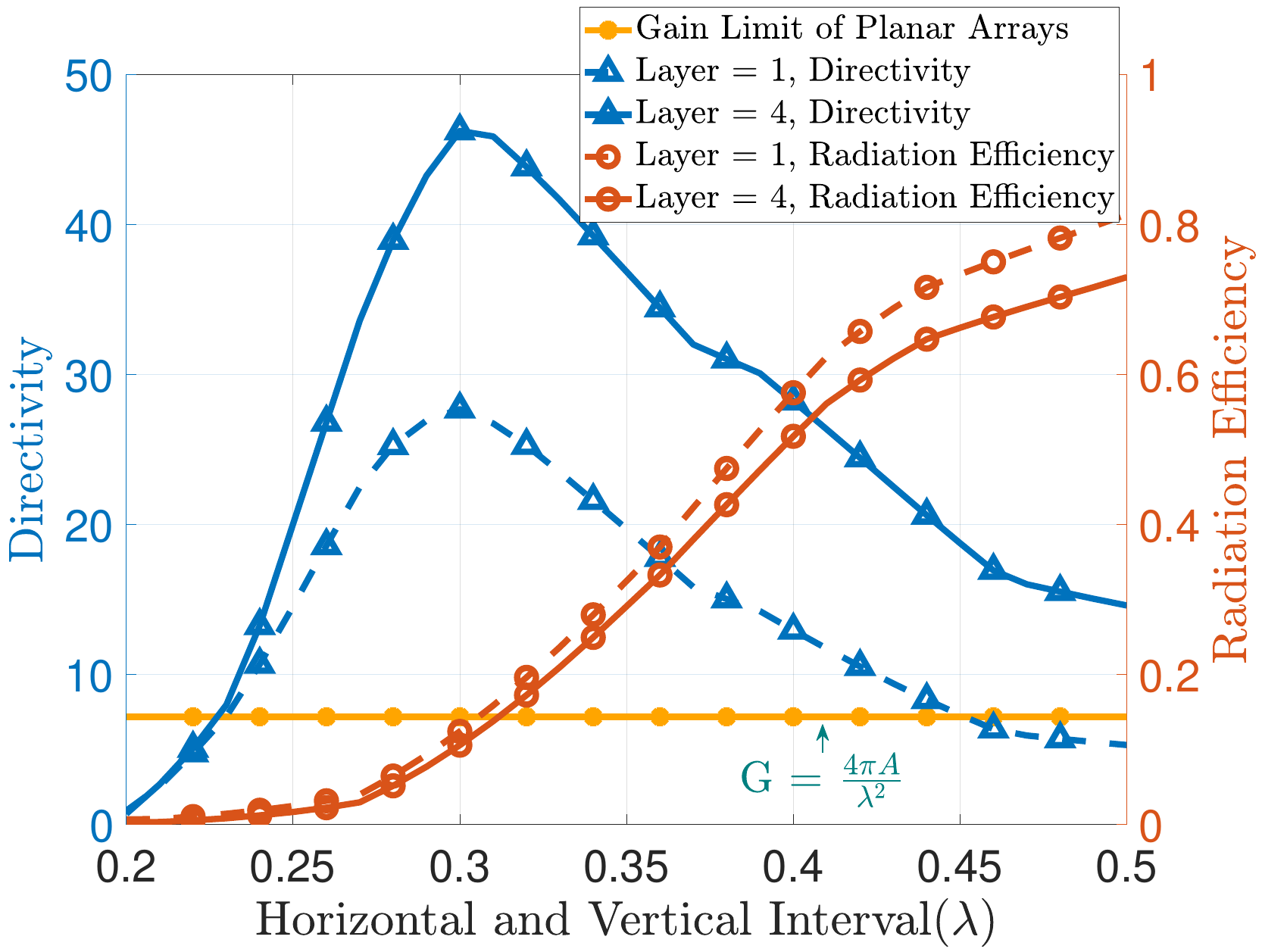}
  \caption{Achievable directivity for fixed aperture EHB system with varying number of antennas.}
  \label{DirectivityvsDistance}
  \vspace{-1ex}
\end{figure}

\par To evaluate the realized gain performance of the 3D array in different antenna spacing cases, in Fig. \ref{RealizedGainvsDistance} we show the achievable realized gain and corresponding radiation efficiency of the proposed holographic array when the SNR is fixed at 20 dB.
As shown in Fig. \ref{RealizedGainvsDistance}, the realized gain of planar array approaches the gain limit at $0.4\lambda$ through EHB, while our proposed EHB scheme with 3D holographic array structure outperforms the planar gain limitation in regions with relatively high radiation efficiency. 
However, when the antenna spacing is reduced below a threshold, for example, $0.4\lambda$ in this case, the corresponding radiation efficiency decreases rapidly and results in a rapid decrease in realized gain as well. Note that the optimal antenna spacing for maximum gain is different from that of the maximum directivity due to the impact of radiation efficiency. 
\vspace{-2ex}
\begin{figure}[!ht]
  \centering
  \includegraphics[width=3.3in]{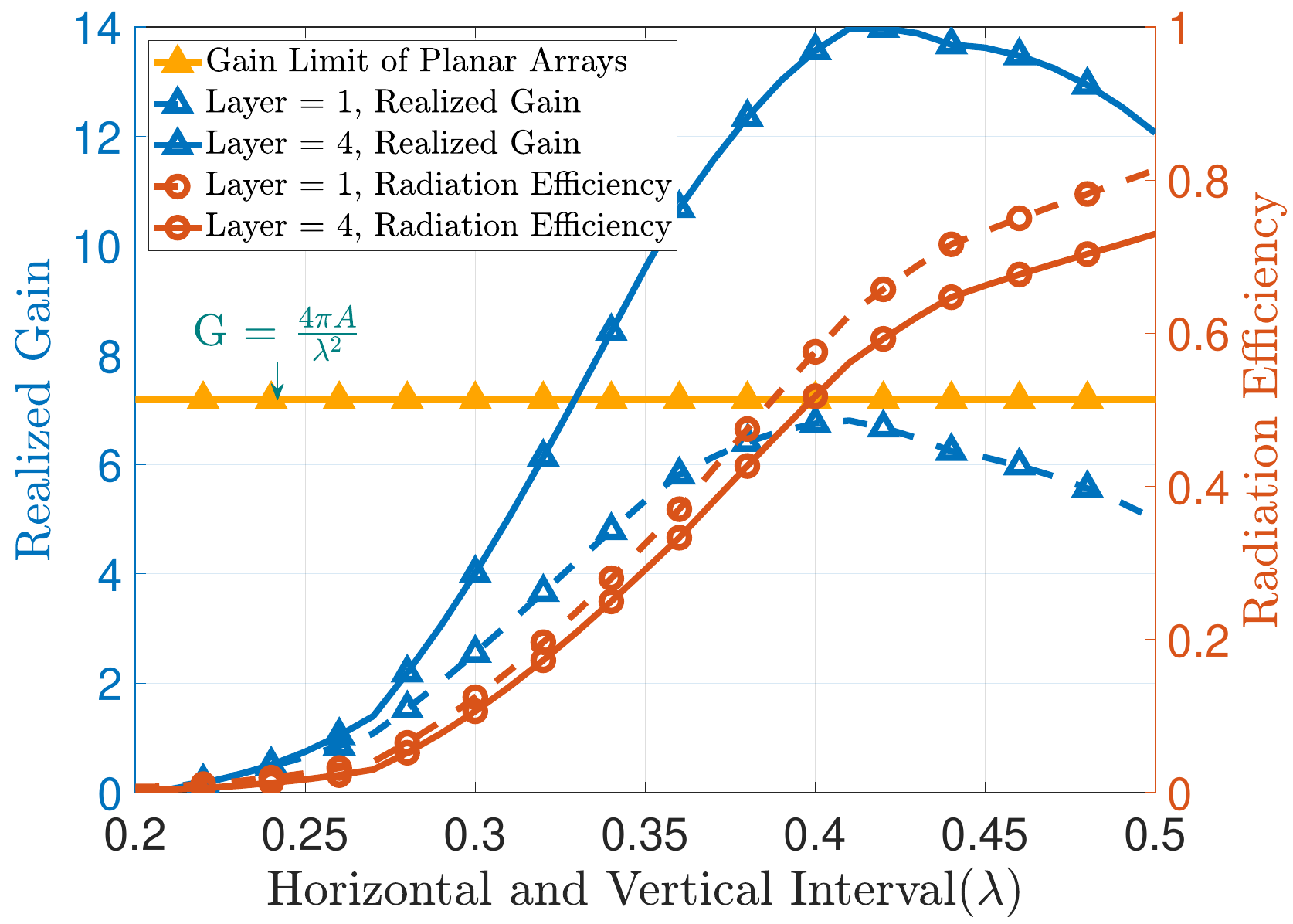}
  \caption{Achievable realized gain for fixed aperture EHB system with varying number of antennas.}
  \label{RealizedGainvsDistance}
\end{figure}

\par Figs. \ref{DirectivityvsDistance} and \ref{RealizedGainvsDistance} demonstrate the superdirective performance of our proposed EHB scheme based on  3D holographic antenna array within a fixed aperture area. In the following figures, we present simulation results for a fixed antenna number, implying that as the antenna spacing varies, the array aperture area also changes accordingly.

\par In Fig. \ref{DirectivityvsDistanceFixedAntennanumber}, we present the directivity and realized gain performance of the proposed EHB scheme for 4 antenna element configurations with different antenna spacings while keeping the SNR fixed at 20 dB. It is evident that for planar densely placed linear array, the directivity increases and then gradually stabilizes at $N_T^2$ as the antenna spacing decreases, which aligns with previous research findings \cite{LHanSuperdirectiveConferene,LHanSuperdirectiveJournal}.
Furthermore, the realized gain of planar array can also approach the gain limit around $0.4\lambda$ spacing. However, note that both the directivity and realized gain achieved by the proposed 3D holographic array surpasses that of the planar linear array, particularly in the region where the antenna element spacing is larger than the aforementioned threshold. This highlights the superior directional and gain performance of the proposed 3D structure compared to its traditional counterpart, emphasizing its potential for practical applications.


\begin{figure}[!ht]
  \vspace{-2ex}
  \centering
  \includegraphics[width=3.3in]{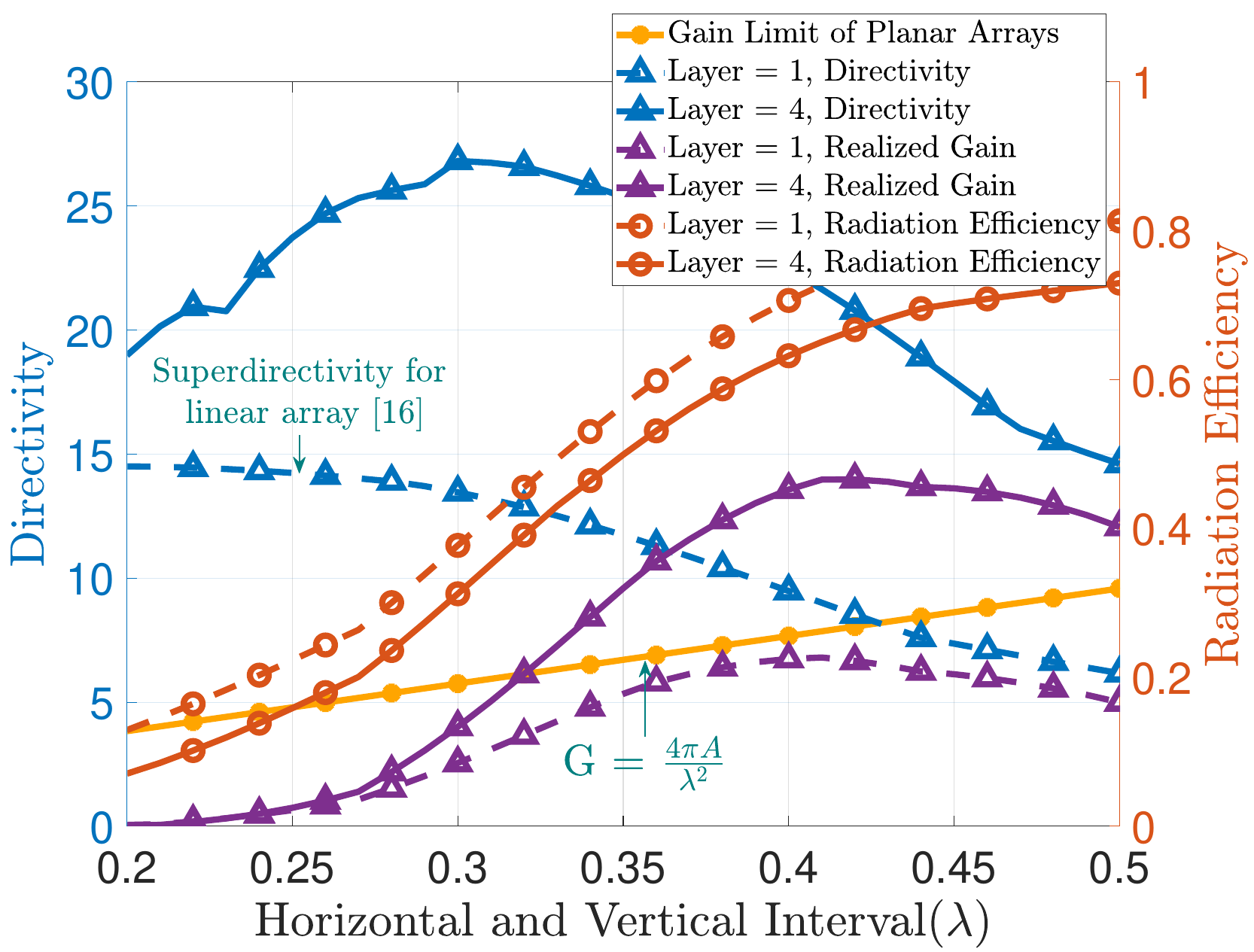}
  \caption{Achievable directivity for fixed antenna number EHB system with varying array aperture.}
  \label{DirectivityvsDistanceFixedAntennanumber}
  \vspace{-0.5ex}
\end{figure}

\par At last, we show the achievable communication rate for a fixed antenna number with different spacings in Fig \ref{SumRatevsDistanceFixedAntennanumber}. 
The results further demonstrate that it is necessary to apply EHB to fully utilize the 3D holographic antenna array, otherwise, there is only marginal benefit in expanding the 3D array.
This emphasizes the crucial role of the proposed beamforming techniques in harnessing the full potential of the 3D holographic array structure and maximizing the sum rate performance. 
Additionally, it can be observed that the communication rate exhibits an upward and then downward trend as the antenna spacing increases from $0.2\lambda$ to $0.5\lambda$, which is in coherence with the realized gain trend in Fig. \ref{DirectivityvsDistanceFixedAntennanumber}. The reason for this phenomenon is as follows: in the sub-half-wavelength region, as the antenna spacing increases, the mutual coupling effect decreases, resulting in reduced superdirectivity but higher radiation efficiency. Their tradeoff results in the fluctuation in this region. 
In the region above the half-wavelength, however, a larger antenna spacing corresponds to a larger aperture area under the constraint of a fixed number of antennas. This leads to a naturally higher array gain and communication rate. In this region, the impact of mutual coupling becomes less significant compared to the increased aperture area, and the fluctuation is due to sidelobes. 

\begin{figure}[!ht]
  \centering
  \includegraphics[width=3.3in]{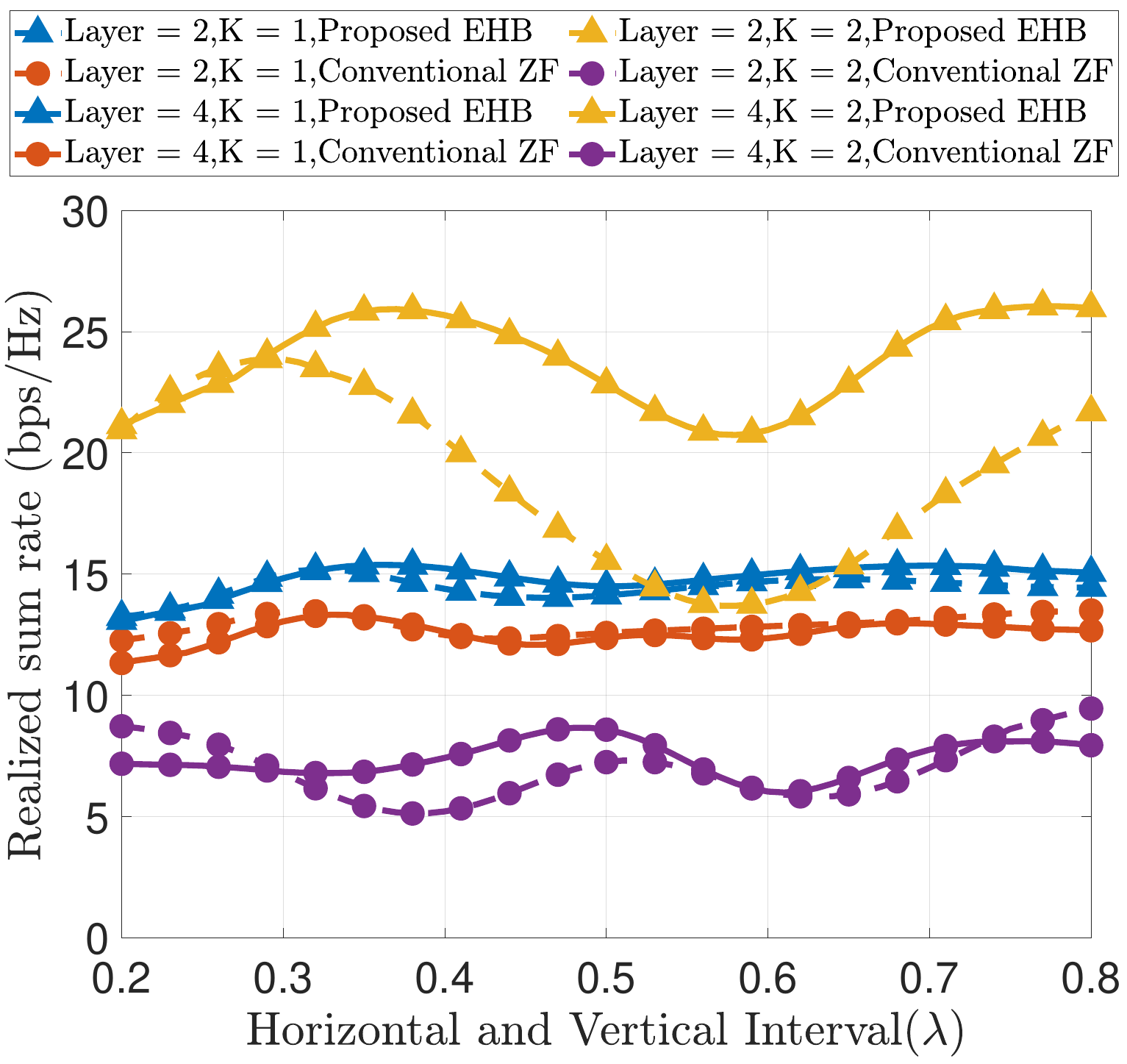}
  \caption{Achievable sum rate for fixed antenna number with varying array aperture.}
  \label{SumRatevsDistanceFixedAntennanumber}
  \vspace{-8ex}
\end{figure}

\section{Conclusion}
In this paper, we proposed an EHB communication scheme based on 3D holographic antenna arrays that takes the impact of array radiation patterns into account,  and incorporates the coupling effect of antenna arrays to have the superdirective gain. 
Specifically, the implementation of analog precoding involves the real-time adjustment of the radiation pattern to adapt to the wireless environment. Meanwhile, the digital precoding is optimized based on the channel characteristics of analog precoding. 
As a result, our approach enables the attainment of programmable spatial patterns and a relatively flat beamforming gain, facilitating the realization of multiple superdirective beams in arbitrary directions. Through simulations, we demonstrate the effectiveness of our proposed method. Additionally, we identify the presence of the optimal antenna spacing in the array geometry, considering the trade-off between the radiation efficiency and superdirective gain.

\vspace{-2ex}
\begin{appendices} 
\section{Proof of theorem 1}
  Firstly we give the Lagrangian function of (\ref{P5}) without the rank-one constraint as follows:
  \vspace{-2ex}
  \begin{align}
      & F(t_k,b_k,\lambda_1,\lambda_{2,k},\lambda_{3,k},\boldsymbol{\lambda}_{4,k},\boldsymbol{F}_k) = \sum \limits_{k=1}^{K} \log(1+t_k) \notag \\
      & + \lambda_1 \left(P_{T,\text{digital}} - \sum \limits_{k=1}^{K}{\rm Tr}(\boldsymbol{F}_k)\right) \notag\\
      & + \sum \limits_{k=1}^{K} \lambda_{2,k} \left[u_k - \left(\sum \limits_{i \neq k}{\rm Tr}(\boldsymbol{H}_k\boldsymbol{F}_i) + \sigma^2\right)\right] \notag \\
      & + \sum \limits_{k=1}^{K} \lambda_{3,k} \left[{\rm Tr}(\boldsymbol{H}_k\boldsymbol{F}_k) - \left(\frac{t_k^{(n)}}{2u_k^{(n)}}u^2_k + \frac{u^{(n)}_k}{2t^{(n)}_k}t^2_k\right)\right] \notag\\
      & + \sum \limits_{k=1}^{K} {\rm Tr}(\boldsymbol{\lambda}_{4,k} \boldsymbol{F}_k),
  \end{align}
  where $\lambda_1$ - $\boldsymbol{\lambda}_4$ represent the Lagrange multipliers corresponding to constraints in (\ref{P5}). Since the relaxed SDP problem $P5$ is convex, Slater's condition is satisfied and the gap between the primal problem and its dual problem is zero \cite{SlaterCondition}. Therefore, the Karush-Kuhn-Tucker (KKT) conditions are necessary and sufficient for the optimal solution of $P5$ without the rank-one constraint. The detailed KKT conditions related to the optimal digital beamforming auxiliary matrix $\boldsymbol{F}_k^*$ can be written as
  \begin{subequations}
    \begin{align}
      \lambda_1^*\boldsymbol{I} + \sum \limits_{i \neq k}\lambda_{2,i}^*\boldsymbol{H}_i - \lambda_{3,k}^*\boldsymbol{H}_k = \boldsymbol{\lambda}_{4,k}^* & \\
      \boldsymbol{\lambda}_{4,k}^*\boldsymbol{F}_k = \boldsymbol{0} & \label{KKT2}\\ 
      \boldsymbol{\lambda}_{4,k}^*\succeq \boldsymbol{0},
    \end{align}
  \end{subequations}
where $\lambda_1^*$ - $\boldsymbol{\lambda}_4^*$ are the optimal Lagrange multipliers. Since $\boldsymbol{I}$ is a full-rank matrix, $\lambda_1^* \geq 0$ and $\lambda_{2,k}^* \geq 0$, we define $\boldsymbol{Y} = \lambda_1^*\boldsymbol{I} + \sum \limits_{i \neq k}\lambda_{2,i}^*\boldsymbol{H}_i$, which is a positive-definite matrix with full rank. Based on this, we have
\begin{equation}
  \begin{aligned}
    {\rm rank}(\boldsymbol{\lambda}_{4,k}) &= {\rm rank}(\boldsymbol{Y} - \lambda_{3,k}^*\boldsymbol{H}_k) \\
    & \geq {\rm rank}(\boldsymbol{Y}) - {\rm rank}(\lambda_{3,k}^*\boldsymbol{\hat{h}}_k\boldsymbol{\hat{h}}_k^H)\\ 
    &\geq {\rm rank}(\boldsymbol{Y}) - 1,
  \end{aligned}
\end{equation}
where $\boldsymbol{\hat{h}}_k = \boldsymbol{C}^*{\rm diag}(\boldsymbol{X}^*(:,k))\boldsymbol{i}^*$. Therefore, we can claim that the rank of $\boldsymbol{\lambda}_{4,k}^*$ is either $N_T$ or $N_T -1$. If ${\rm rank}(\boldsymbol{\lambda}_{4,k}^*)=N_T$, then the optimal $\boldsymbol{F}_k^*$ should be zero matrix according to (\ref{KKT2}), which means that the digital matrix is zero and no signal is transmitted. Thus, the feasible solution should be ${\rm rank}(\boldsymbol{\lambda}_{4,k}^*)=N_T-1$, and the null space of $\boldsymbol{\lambda}_{4,k}^*$ is 1 dimension. From (\ref{KKT2}), the derived optimal beamforming auxiliary matrix $\boldsymbol{F}_k^*$ must lie in the null space of $\boldsymbol{\lambda}_{4,k}^*$, resulting in a rank-one solution, which fulfills the initial constraint we dropped. The same can be proven for $P6$.
\end{appendices}

\bibliographystyle{IEEEtran}
\bibliography{IEEEabrv,references} 

\end{document}